\documentclass[10pt]{amsart}

\usepackage{amsmath,amssymb,amsfonts,latexsym,epsfig}

\usepackage{bbm,vector}

\usepackage{fullpage}

\usepackage[english, algosection, algoruled, noline]{algorithm2e}
\usepackage{hyperref}

\usepackage[numbers,sort]{natbib}

\usepackage{url}

\newcommand{\EXAMPLE}[1]{%
  \begin{center}\fbox{\begin{minipage}{0.95\textwidth}{\sf{}#1}\end{minipage}}%
  \end{center}}

\newtheorem{theorem}{Theorem}[section]

\newtheorem{definition}[theorem]{Definition}
\newtheorem{proposition}[theorem]{Proposition}

\def\kk{\mathbb{K}}
\def\CC{\mathbb{C}}
\def\NN{\mathbb{N}}
\def\RR{\mathbb{R}}

\def\OO{\mathcal{O}}
\def\Sc{\mathcal{S}}

\def\x{\mathbf{x}}
\def\db{\mathbf{d}}
\def\hb{\mathbf{h}}

\def\Proj{\mathbb{P}}
\def\Ac{\mathcal{A}}
\def\AA{\mathcal{A}_{\Lambda}}
\def\<{\langle}
\def\>{\rangle}

\newcommand{\dual}[1]{{#1}^{*}}
\def\dAA{\dual{\mathcal{A}}_{\Lambda}}
\def\dRR{\dual{ R}}

\def\Hom{\operatorname{Hom}}
\def\rank{\operatorname{rank}}

\def\mat#1{\mathbbmss#1}
\def\form#1{\mathcal#1}
\def\abs#1{\mathopen| #1 \mathclose|}

\def\set#1{\{ #1 \}}
\def\Set#1{\left\{ #1 \right\}}
\def\Bigbar#1{\mathrel{\left|\vphantom{#1}\right.}}
\def\Setbar#1#2{\Set{#1 \Bigbar{#1 #2} #2}}

\title{Symmetric tensor decomposition}
\author{Jerome Brachat
\and Pierre Comon
  \and Bernard Mourrain
  \and Elias P. Tsigaridas
}

\begin{document}
\maketitle

\begin{abstract}
We present an algorithm for decomposing a symmetric tensor, of dimension n
and order d as a sum of rank-1 symmetric tensors, extending the algorithm of
Sylvester devised in 1886 for binary forms. 

We recall the correspondence between the decomposition of a homogeneous
polynomial in n variables of total degree d as a sum of powers of linear
forms (Waring's problem), incidence properties on secant varieties of the
Veronese Variety and the representation of linear forms as a
linear combination of evaluations at distinct points. Then we reformulate
Sylvester's approach from the dual point of view.

Exploiting this duality, we propose necessary and sufficient conditions for
the existence of such a decomposition of a given rank, using the properties of
Hankel (and quasi-Hankel) matrices, derived from multivariate polynomials and
normal form computations. This leads to the resolution of polynomial
equations of small degree in non-generic cases.

We propose a new algorithm for symmetric tensor decomposition, based on
this characterization and on linear algebra computations with these Hankel
matrices. 

The impact of this contribution is two-fold. First it permits an efficient
computation of the decomposition of any tensor of sub-generic rank, as
opposed to widely used iterative algorithms with unproved global convergence
(e.g. Alternate Least Squares or gradient descents). Second, it gives tools
for understanding uniqueness conditions, and for detecting the rank.
\end{abstract}

\tableofcontents
\newpage

\section{Introduction}
Symmetric tensors show up in applications mainly as high-order derivatives of multivariate
functions. For instance in Statistics, cumulant tensors are derivatives of the second
characteristic function \cite{Mccu87}. 

Tensors have been widely utilized in Electrical Engineering since the nineties
\cite{SwamGS94:ieeesp}, and in particular in Antenna Array Processing
\cite{DogaM95a:ieeesp} \cite{ChevAFC05:ieeesp} or Telecommunications
\cite{VandP96:ieeesp} \cite{Chev99:SP} \cite{SidiGB00:ieeesp} \cite{FerrC00:ieeesp}
\cite{DelaC07:SP}.  
Even earlier in the seventies, tensors have been used in Chemometrics  \cite{Bro97:cils}
or Psychometrics \cite{KierK91:psy}. 

Arithmetic complexity is also an important field where the understanding of tensor
decompositions has made a lot of progress, especially third order tensors, which represent
bilinear maps  \cite{Krus77:laa}  \cite{BiniCRL79:ipl} \cite{Stra83:laa}
\cite{Land08:bams}. 

Another important application field is Data Analysis. For instance, Independent Component
Analysis, originally introduced for symmetric tensors whose rank did not exceed dimension
\cite{Como92:elsevier} \cite{Card98:procieee}. Now, it has become possible to estimate
more factors than the dimension \cite{Dono01:ieeeit} \cite{JianS04:ieeesp}. In some
applications, tensors may be symmetric only in some modes \cite{ComoR06:SP}, or may not be
symmetric nor have equal dimensions \cite{CichA02} \cite{SmilBG04}. 
Further numerous applications of tensor decompositions may be found in \cite{CichA02} \cite{SmilBG04}.

Note that in some cases, tensors are encountered in the form of a collection of symmetric
matrices \cite{Dele78:psy} \cite{Flur84:jasa} \cite{VandP96:ieeesp} \cite{PhamC01:ieeesp}
\cite{TenbSR04:laa}, in which case they may enjoy symmetries in some modes but not in
others. Conversely, some algorithms treat symmetric tensors as a collection of symmetric
matrix slices \cite{Yere02:ieeesp} \cite{ZiehNM04:jmlr} \cite{DelaCC07:ieeesp}. 

The problem of decomposition of a symmetric tensor, that we consider in this
paper, is a rank determinant problem which extends the Singular Value Decomposition (SVD)
problem for symmetric matrices. This former method is an important tool in (numerical)
linear algebra, which received at lot of attention and which is 
routinely used in many applications \cite{gvl-book-1983}. As exhibited above, the
extension to general symmetric tensors also appears in many application
domains. However, many issues either theoretical or algorithmic remains
unsolved. Among solved problems, let us mention the determination of the
minimal number of terms in the decomposition of a generic tensor
\cite{AlHi92a}, which is stated there in terms of a dual interpolation
problem. See \cite[chap. 2]{ia-book-1999} and section~\ref{sec:prob-formulation} for the
link between these two points of view. Among open problems are the
determination of the maximal rank of tensors of given degree and
dimension, or the determination of the stratification of the set of symmetric tensors
by the rank. See however \cite{cm-binary-arxiv-2001} for an answer in the
binary case. For a detailed presentation of the symmetric tensor
decomposition problem, from a projective algebraic geometric point of view, we refer
to \cite{ia-book-1999}. The properties of so-called catalecticant matrices, related to
the apolar duality induced by the symmetric tensor associated with homogeneous polynomials
of a given degree, are extensively studied. 

In a disconnected way, probably because of language barrier, investigations
of this problem in numerical analysis have been developed, inspired by the
successful work on order 2 tensors, i.e. matrices. However, despite their obvious
practical interest, numerical algorithms presently used in most scientific
communities are suboptimal, in the sense that they either do not fully
exploit symmetries \cite{AlbeFCC04:laa}, minimize different successive
criteria sequentially \cite{Yere02:ieeesp} \cite{DelaCC07:ieeesp}, or are
iterative and lack a guarantee of global convergence \cite{HarsL94:csda}
\cite{Paat99:jcgs}. In addition, they often request the rank to be much
smaller than generic.  Among these popular methods, we refer to PARAFAC
techniques \cite{Bro97:cils}, heavily applied for solving an ill-posed
problem\dots\ Indeed contrarily to the matrix case, the set of symmetric tensors
of rank $\le r$ is not closed, and its closure has singularities
corresponding to tensors of rank $>r$. This explains why iterative numerical
methods encounter difficulties to compute a tensor decomposition. 
For more details on open problems on symmetric tensors, see \cite{cglm-simax-2008}.

The goal of this paper is to describe a new algorithm able to decompose a symmetric
tensor of arbitrary order and dimension into a sum of rank-one terms. 
The algorithm proposed in this paper is inspired from Sylvester's theorem \cite{KungR84:ams},
and extends its principle to larger dimensions. 
Using apolar duality on polynomials, we show that the symmetric tensor
decomposition problem reduces to the decomposition of a linear form as a
linear combination of evaluations at distinct points. 
We give a necessary and sufficient condition for the existence of a 
decomposition of rank $r$, based on rank conditions of Hankel operators
or commutation properties.
Instead of working, degree by degree, as in \cite{ia-book-1999}, we consider
affine situations in order to treat at the same time the various
homogeneous components. 
In the binary case, the decomposition problem can be solved directly by
computing ranks of catalecticant. In higher dimension, this not so simple.
An extension step is required to find the decomposition.
This leads to the resolution of a polynomial system of small
degree, from which we deduce the decomposition by solving a simple eigenvalue problem,
thanks to linear algebra manipulations.

The algorithm is not restricted to strictly sub-generic ranks as for the method
proposed in \cite{ia-book-1999}[chap. 5]. 
In sub-generic cases, the decomposition is essentially unique (i.e. up to scale
and permutation) when some rank conditions are satisfied. 
Our algorithm fully exploits this symmetry 
and provides a complete answer to the questions of uniqueness and
computation, for any order \cite{cm-sp-1996}.

In the following section, we recall the method deduced from  Sylvester's
theorem to decompose a binary form. In section \ref{sec:prob-formulation}, we give three equivalent
formulations of the same problem, used and studied in different communities.
In section \ref{sec:4}, we develop the duality point of view, extending the
notion of generalized additive decomposition, introduced in
\cite{ia-book-1999}, to any dimension. 
Section \ref{sec:5} is devoted to the algebraic characterization of the
extension property of linear forms, in terms of rank condition on
multivariate Hankel operators, or on commutation properties.
Finally in section \ref{sec:6}, we describe the algorithm and give examples.

\subsection{The binary case}

The present contribution is a generalization of Sylvester's algorithm devised
to decompose homogeneous polynomials in two variables into a sum of powers of
linear forms \cite{sylv-cr-1886} \cite{cm-binary-arxiv-2001}. 
It is hence convenient to first recall the latter algorithm.

\begin{theorem}[Sylvester, 1886]
  A binary quantic $p(x_1,x_2) = \sum_{i=0}^d {d \choose i} \, c_i \, 
  x_1^i \, x_2^{d-i}$ can be written as a sum of $d^{\mathrm{th}}$
  powers of $r$ distinct linear forms in $\CC$ as:
  \begin{equation}\label{cand:eq}
    p(x_1,x_2) = \sum_{j=1}^r \lambda_j \: (\alpha_j \, x_1 + \beta_j \, x_2)^d,
  \end{equation}
  if and only if {\bf(i)}~there exists a vector $\bvec{q}$ of dimension $r+1$, with
  components $q_\ell$, such that
  \begin{equation}\label{eq:sylves}
    \left[\begin{array}{cccc}
        c_0 & c_1 & \cdots & c_r\\
        \vdots   &  &    &  \vdots\\
        c_{d-r} & \cdots & c_{d-1} & c_d
      \end{array}\right] \: \bvec{q} = {\bf 0}.
  \end{equation}
  and {\bf(ii)}~the polynomial $q(x_1,x_2) = \sum_{\ell=0}^r q_\ell\,
  x_1^\ell\, x_2^{r-\ell}$ admits $r$ distinct roots, i.e. it can be written as
  $q(x_1,x_2)= \prod_{j=1}^{r} (\beta_{j}^*\,x_1 - \alpha_{j}^*\, x_2)$.
\end{theorem}

The proof of this theorem is constructive 
\cite{sylv-cr-1886} \cite{cm-sp-1996} \cite{cglm-simax-2008} and yields Algorithm \ref{algo:sylvester}.
Given a binary polynomial $p(x_1,x_2)$
of degree $d$ with coefficients $a_i={d \choose i}\,c_i$, $0\le i\le d$,
define the Hankel matrix $H[r]$ of dimensions $d-r+1\times r+1$ with entries
$H[r]_{ij}=c_{i+j-2}$. 
\begin{algorithm}\caption{\textsc{Binary form decomposition}}\label{algo:sylvester}
\KwIn{A binary form $p(x_{0},x_{1})$ of degree $d$.}
\KwOut{A decomposition of $p$ as  $p(x_1,x_2) = \sum_{j=1}^r \lambda_j \:\bvec{k}_j(\bvec{x})^d$ with $r$ minimal.}
\begin{enumerate}
  \addtolength{\itemsep}{-0.8\parsep minus 0.8\parsep}
\item
  Initialize $r=0$ 
\item\label{step2-item}
  Increment $r\leftarrow r+1$
\item
  If the column rank of $H[r]$ is full, then go to step \ref{step2-item}
\item
  Else compute a basis $\{\bvec{k}_{1},\ldots, \bvec{k}_{l}\}$ of the right
kernel of $H[r]$.
\item\label{step5-item}
  Specialization:
  \begin{itemize}
    \addtolength{\itemsep}{-0.8\parsep minus 0.9\parsep}
  \item Take a generic vector $\bvec{q}$ in the kernel, e.g.
    $\bvec{q}=\sum_i \mu_i \bvec{k}_{i}$
  \item
    Compute the roots of the associated polynomial   
    $q(x_1,x_2)=\sum_{\ell=0}^r q_\ell \, x_1^\ell\,x_2^{d-\ell}$. 
    Denote them ($\beta_j,-\alpha_j)$, where $|\alpha_j|^2+|\beta_j|^2=1$.
  \item
    If the roots are not distinct in $\Proj^2$, try another specialization. If
    distinct roots cannot be obtained, go to step \ref{step2-item}. 
  \item
    Else if $q(x_1,x_2)$ admits $r$ distinct roots then compute coefficients
    $\lambda_j$, $1\le j\le r$, by solving the linear system below, where $a_i$ denotes ${d \choose i}\,c_i$
    $$
    \left[\begin{array}{ccc}
        \alpha_1^d & \dots & \alpha_r^d\\
        \alpha_1^{d-1}\beta_1 & \dots & \alpha_r^{d-1}\beta_r\\
        \alpha_1^{d-2}\beta_1^2 & \dots & \alpha_r^{d-1}\beta_r^2\\
        : & : & :\\
        \beta_1^d & \dots & \beta_r^d
      \end{array}\right] \, \bvec{\lambda} = 
    \left[\begin{array}{c}
        a_0\\ a_1\\ a_2\\ :\\ a_d
      \end{array}\right]
    $$
  \end{itemize}
\item
  The decomposition is 
  $p(x_1,x_2) = \sum_{j=1}^r \lambda_j \:\bvec{k}_j(\bvec{x})^d$, 
  where $\bvec{k}_j(\bvec{x})=(\alpha_j\,x_1 + \beta_j\,x_2)$.
\end{enumerate}
\end{algorithm}
Note that step \ref{step5-item} is a specialization only if the dimension of
the right kernel is strictly larger than 1.

\subsection{Notation and preliminaries} 

Let $\kk$ be an algebraically closed field (e.g. $\kk=\CC$ the field of
complex numbers).
For a vector space $E$, its associated projective space is denoted
$\Proj(E)$. For $\bvec{v}\in E-\{0\}$ its class in $\Proj(E)$ is denoted
$\overline{\bvec{v}}$.
Let $\Proj^n$ be the projective space of the field $\kk^n$.

If $\bvec{a} = (a_1, \dots, a_n)$ is a vector in $\NN^{n}$, 
then $|\bvec{a}|$ is the sum of its
elements, i.e. $|\bvec{a}| = \sum_{i=a}^{n}{a_i}$.
We also use the greek letters $\alpha$  and $\beta$ for vectors in $\NN^n$.
We denote by $\bvec{x}^{\bvec{\alpha}}$ the monomial $x_1^{\alpha_1}\, \cdots \, x_n^{\alpha_n}$.
For a set $B = \set{b_1, \dots, b_m}$, we denote by $\<B\>$, respectively $(B)$, 
the corresponding vector space, resp. ideal.

Let $R$ be the ring of polynomials $\kk[x_1, \dots, x_n]$,
while $R_d$ will denote the ring of polynomials of (total) degree at most $d$.
The set 
$\set{ \bvec{x}^{\bvec{\alpha}}}_{|\bvec{\alpha}| \leq d} = \set{ x_1^{\alpha_1} \cdots x_n^{\alpha_n}}_{\alpha_1 + \cdots + \alpha_n \leq d}$
represents the elements of the monomial basis of the vector space $R_d$
and contains ${n+d \choose d}$ elements.
Hereafter, the superscript $h$ denotes the homogenization of a polynomial. 
We denote by $S_{d}$ the vector space of homogeneous polynomials 
in $n+1$ variables  $x_0, x_1, \dots, x_n$.
This is also the symmetric $d$-th power $S^{d}(E)$ where
$E=\<x_{0},\ldots,x_{n}\>$. The dehomogenization of o polynomial $f\in
S_{d}$ with respect to the variable $x_{0}$ is denoted $f^{a}:=
f(1,x_{1},\ldots, x_{n})$.

Duality is an important ingredient of our approach. For a comprehensive
treatment of duality on multivariate polynomials, we refer the reader 
to \cite{mp-jcomplexity-2000}.  Hereafter, for a $\kk$-vector space $E$
its dual $\dual{E}=\Hom_{\kk}(E,\kk)$ is the set of $\kk$-linear forms
form $E$ to $\kk$.

A basis of the dual space $\dual{R}_d$,
is the set of linear forms that compute the
coefficients of a polynomial in the primal basis. It is 
denoted by $\set{ \bvec{d}^{ \bvec{\alpha}}}_{ \abs{ \bvec{\alpha}} \leq d}$.

We may identify $\dRR$ with the (vector) space of formal power
series, i.e. $\kk[[ \bvec{d}]] = \kk[[ d_1, \dots, d_n]]$. 
Any element $\Lambda \in \dual{R}$ can be decomposed as 
$\Lambda = \sum_{\bvec{a}} \Lambda (\bvec{x}^{\bvec{a}})\bvec{d}^{\bvec{a}}.$
Typical elements of $\dual{R}$ are the linear forms that correspond to the
evaluation at a point  $\zeta \in \kk^n$:
\begin{displaymath}
  \begin{array}{lcl}
    \mat{1}_{\zeta }  & : & R \rightarrow \kk \\
    &  &  p \mapsto p (\zeta) \\
  \end{array}
\end{displaymath}
The decomposition of $\mat{1}_{\zeta}$ in the basis $\set{ \bvec{d}^{\bvec{a}}}_{ \abs{ \bvec{a}} \leq d}$
is $\mat{1}_{\zeta} = \sum_{\bvec{a}} \zeta^{ \bvec{a}} \bvec{d}^{\bvec{a}}$.
Such an evaluation form can be composed with differentiation.
In fact, if $\theta (\partial_1, \ldots, \partial_n)$ is a differential polynomial, then
\begin{displaymath}
  \begin{array}{lcl}
    \mat{1}_{\zeta} \circ \text{$\theta (\partial_1, \ldots, \partial_n)$} & : &
    R \rightarrow \kk \\
    & & p \mapsto \theta ( \partial_1, \ldots, \partial_n) (p) (\zeta).
  \end{array}
\end{displaymath}
The dual space $\dual{R}$ has a natural structure of $R$-module \cite{EM08} which is
defined as follows: for all $p \in R$, and for all $\Lambda \in \dual{R}$
consider the linear operator
\begin{displaymath}
  \begin{array}{lcl}
    p \star \Lambda & : & R \rightarrow \kk \\
    && q \mapsto \Lambda (p q).
  \end{array}
\end{displaymath}
In particular, we have
$x_i \star \bvec{d}^{\bvec{a}} = 
  \left\{ 
    \begin{array}{lcl}
      d_1^{a_1} \cdots d_{i-1}^{a_{i-1}} d_i^{a_i - 1} d_{i + 1}^{a_{i+1}} \cdots d_n^{a_n} & & \text{if }  a_i > 0,\\
      0 && \text{otherwise} .
    \end{array}
  \right.
$

\section{Problem formulations}
\label{sec:2}
\label{sec:prob-formulation}
In this section, we present three different formulations of the same problem,
that we consider in this paper.
\subsection{Polynomial decomposition}
\label{sec:def-poly-decomp}
A symmetric tensor $\mathbb{[}a_{j_{0},\ldots, j_{n}}\mathbb{]}$ of order $d$
and dimension $n$ can be associated with a homogeneous polynomial $f( \x) \in
S_d$:
\begin{equation}
  \label{eq:homo-poly}
  f( \x) = 
  \sum_{j_0 + j_1 + \cdots + j_n = d}{a_{j_0, j_1, \ldots, j_n} x_0^{j_0} x_1^{j_1} \cdots x_n^{j_n}}.
\end{equation}
Our goal is to compute a decomposition of $f$ as a sum of $d^{\mathrm{th}}$ powers of linear
forms, i.e.
\begin{equation}
  \label{eq:poly-decomp}
  f(\x) = \sum_{i=1}^{r}{\lambda_{i}\, ( k_{i,0} x_0 + k_{i,1} x_1 + \cdots + k_{i, n} x_n)^d} = 
  \lambda_{1}\, \bvec{k}_1(\bvec{x})^d + 
  \lambda_{2}\,\bvec{k}_2(\bvec{x})^d + \cdots + 
  \lambda_{r}\,\bvec{k}_r(\bvec{x})^d,
\end{equation}
where $\lambda_{i}\neq 0$, $\bvec{k}_{i}\neq 0$,
and $r$ is the smallest possible.
This minimal $r$  is called the {\em rank} of $f$.

The {\em direct} approach to solve this decomposition problem, is the
following. Consider the  relation
\begin{displaymath}
  f(\x) = \sum_{i=1}^{r}{( k_{i,0} x_0 + k_{i,1} x_1 + \cdots + k_{i, n} x_n)^d},
\end{displaymath}
where $\bvec{k}_{i}\neq 0$.
We assume that  $r$, the rank, is known and the smallest possible.
We consider the  $r(n+1)$ coefficients $k_{i,j}$ of the linear forms as unknowns.
We expand (symbolically) the right hand side of the relation.
The two polynomials of the left and right hand sides are equal, thus by
equating the coefficients of the same monomials we get a polynomial system in
the coefficients $k_{i,j}$.
This is an over-constrained polynomial system of ${n+d \choose d}$ equations and
$r(n+1)$ unknowns. The polynomials of the system are
homogeneous of degree $d$ and the magnitude of their coefficients
is $\leq r {n+d \choose d}$. 
This approach describes the problem of decomposition in a non-optimal
way. It introduces $r!$ redundant solutions, since every permutation of the
linear forms is also a solution. 
Another drawback of this approach is that the polynomials involved are of high
degree, that is, $d$. The reader can compare this with the degree 2
polynomial system,
described in Section \ref{sec:trunc:hankel}, containing the polynomials that
we  have to solve in order to extend the matrix.

In the following sections, we are going to describe a new method, which 
is much more efficient to solve this decomposition problem.

\subsection{Veronese and Secant Varieties}
\label{sec:veronese}
Let us recall the well-known correspondence between the symmetric outer product decomposition
and secant varieties for symmetric tensors. The set of symmetric tensors or homogeneous
polynomials of the form 
$\bvec{k}( \bvec{x})^d = ( k_0 x_0 + k_1 x_1 + \cdots + k_n x_n)^d$
for $\bvec{k} = ( k_0 , k_1, \dots, k_n) \in \kk^n$
is a closed algebraic set. 
Scaling the vector $\bvec{k}$ by a non-zero scalar $\lambda$
yields a homogeneous polynomial scaled by $\lambda^d$.
Thus, we can also consider this construction as a
map $\bvec{k} \mapsto \bvec{k}( \bvec{x})^d$ from the projective space 
$\Proj^{n-1}$  to the projective space of symmetric tensors:
\begin{displaymath}
  \begin{array}{lclcl}
    \nu & : & \Proj( S_{1}) &\rightarrow & \Proj( S_d) \\
    &   & \bvec{k}(\bvec{x})   &\mapsto & \bvec{k}(\bvec{x})^d .
  \end{array}
\end{displaymath}
The image of $\nu$ is called the Veronese variety $\mathcal{V}_{n, d}$
\cite{zak-ams-1993,harris-book}. Following this point of 
view, a tensor is of rank 1 if it corresponds to a point on the Veronese
variety. A tensor is of rank $\leq r$ if it is a linear combination of $r$ tensors
of rank 1. In other words, it is in the linear space spanned by $r$ points of
the Veronese variety. The closure of the $r-$dimensional linear space spanned
by $r$ points of the Veronese variety $\mathcal{V}_{n,d}$ is called the
$(r-1)$-secant variety of $\mathcal{V}_{n,d}$ and denoted
$\Sc^{r-1}(\mathcal{V}_{n,d})$. 
We refer the reader to \cite{zak-ams-1993,harris-book} for examples and general
properties of these algebraic sets. In the non-symmetric case,  the
so-called Segre variety of the projective space of tensors is considered
 instead of the Veronese variety. It corresponds to the set of (possibly
non-symmetric) tensors of rank 1.

For any $f\in S_{d}-\{0\}$, the smallest $r$ such that $\overline{f}\in
\Sc^{r-1}(\mathcal{V}_{n,d})$ is called the {\em typical rank} or {\em border
rank} of $f$ 
\cite{cglm-simax-2008,TenbSR04:laa,BurgCS97}. 

\subsection{Decomposition using duality}
\label{sec:duality}


Let $f, g \in S_d$, where 
$f=\sum_{|\alpha|=d} f_{\alpha} x_{0}^{\alpha_{0}}\cdots x_{n}^{\alpha_{n}}$
and
$g=\sum_{|\alpha|=d} g_{\alpha} x_{0}^{\alpha_{0}}\cdots x_{n}^{\alpha_{n}}$.
We define the apolar inner product on $S_{d}$ as 
\begin{displaymath}
  \<f,g\> = \sum_{|\alpha|=d} f_{\alpha}\, g_{\alpha} {d \choose \alpha_{0},\ldots,\alpha_{n}}^{-1}.
\end{displaymath}
Using this non-degenerate inner product, we can associate an element of $S_{d}$
with an element $\dual{S}_{d}$, through the following map:
\begin{eqnarray*}
  \tau : S_{d} & \rightarrow \dual{S}_{d}\\
  f & \mapsto \Lambda_f ,
\end{eqnarray*}
where the linear form $\Lambda_f$ is defined as $\Lambda_f: g \mapsto \<f,g\>$.
A simple calculation shows that 
$\<f,\bvec{k}(\bvec{x})^{d}\> = f(\bvec{k})$ 
so that under this duality it holds that 
$\tau( \bvec{k}(\bvec{x})^{d})=\mat{1}_{\bvec{k}} \in \dual{S}_{d}$. 
Moreover, under $\tau$, the polynomial 
$f=\sum_{|\alpha|=d}\, c_{\alpha}\, {d \choose \alpha}\, \x^{\alpha} \in S_{d}$ 
is mapped to 
$\dual{f}= \sum_{|\alpha|=d}\, c_{\alpha}\, \db^{\alpha} \in \dual{S}_{d}$.


The problem of decomposition of $f$ can then be restated as follows:
\begin{quote}
   Given $\dual{f} \in \dual{S}_{d}$, find the minimal number
  of non-zero vectors
  $\bvec{k_{1}}, \ldots, \bvec{k_{r}}\in \kk^{n+1}$ and non-zero scalars $\lambda_{1},
  \ldots, \lambda_{r} \in \kk-\{0\}$ such that 
  $$ 
  \dual{f} = \sum_{i=1}^{r} \lambda_{i} \, \mat{1}_{\bvec{k}_{i}}.
  $$
\end{quote}
By scaling $\bvec{k_{i}}$ and multiplying $\lambda_{i}$ by the inverse of the
$d^{th}$ power of this scaling factor, we may assume that the first non-zero
coordinate of $\bvec{k}_{i}$  is $1$.

\begin{definition} 
  We say that $\dual{f}$ is an affine decomposition if 
  for every $\bvec{k}_{i}$ in the decomposition,   $\bvec{k}_{i,0}\neq 0$.
\end{definition}
By a generic change of coordinates, any decomposition of $\dual{f}$ can be
transformed into an affine decomposition.
To any $\dual{f} \in \dual{S}_{d}$, 
we can associate an element in $\dual{R}_{d}$,
defined by 
$\tilde{\Lambda}_{f}: p \in R_{d} \mapsto \dual{f}(p^{h})$,
where $p^{h}$ is the homogenization in degree $d$ of $p$. 
If $\dual{f}$ admits an affine decomposition with $\bvec{k}_{i,0}=1$
then we also have that $\tilde{\Lambda}_{f}$ coincides
with the linear form
$$ 
\tilde{\Lambda} = \sum_{i=1}^{r} \lambda_{i} \, \mat{1}_{\tilde{\bvec{k}}_{i}}
$$
up to degree $d$, 
where $\tilde{\bvec{k}}_{i}$ is the vector made of the last $n$ coordinates
of $\bvec{k}_{i}$.

\section{Hankel operators and quotient algebra}
\label{sec:4}\label{sec:moments}

In this section, we recall the algebraic tools we will need to describe and
analyze our algorithm.

\label{sec:hankel-op}
\label{sec:quadratic-form}

For any $\Lambda \in \dual{R}$, 
we define the bilinear form $Q_{\Lambda}$, such that 
\begin{displaymath}
  \begin{array}{lcl}
    Q_{\Lambda} & : & R \times R \rightarrow \kk  \\
    &  & (a, b) \mapsto \Lambda (a b).
  \end{array}
\end{displaymath}
The matrix of $Q_{\Lambda}$ in the monomial basis, 
of $R$ is 
$\mat{Q}_{\Lambda} = (\Lambda (\x^{\alpha + \beta}))_{\alpha, \beta}$,
where $\alpha, \beta \in \NN^n$.

For any $\Lambda \in \dual{R}$, we define the Hankel operator $H_{\Lambda}$
from $R$ to $\dual{R}$ as 
\begin{displaymath}
  \begin{array}{lcl}
    H_{\Lambda} & : & R \rightarrow \dRR  \\
    && p \mapsto p \star \Lambda.
  \end{array}
\end{displaymath}
The matrix of the linear operator $ H_{\Lambda}$ 
in the monomial basis, 
and in the dual basis, $\set{ \bvec{d}^{\alpha}}$,
is $\mat{H}_{\Lambda} = ( \Lambda( \x^{\alpha + \beta}))_{\alpha, \beta}$,
where $\alpha, \beta \in \NN^n$.
The following relates the Hankel operators with the bilinear forms.
For all $a, b \in R$, thanks to the $R$-module structure, it holds 
\begin{displaymath}
  Q_{\Lambda} (a, b) = \Lambda( a b) = a \cdot \Lambda (b) = b
  \cdot \Lambda (a) = H_{\Lambda} (a) (b) = H_{\Lambda} (b) (a).
\end{displaymath}
In what follows we will identify $H_{\Lambda}$ and $Q_{\Lambda}$.
\begin{definition} 
  Given $B=\{b_{1},\ldots,b_{r}\}, B' =\{b'_{1},\ldots, b'_{r'}\}\subset R$ we define
  $$ 
  H^{B,B'}_{\Lambda} : \<B\> \rightarrow \dual{ \<B'\>},
  $$
  as the restriction of $H_{\Lambda}$ to the vector space $\<B\>$ and
inclusion of $R^{*}$ in $\dual{ \<B'\>}$. Let $\mat{H}^{B,B'}_{\Lambda} = 
  (\Lambda(b_{i}\, b'_{j}))_{1 \le i\le r, 1\le j\le r'}$. If  $B'=B$, we
also use the notation $H^{B}_{\Lambda}$ and $\mat{H}^{B}_{\Lambda}$.
\end{definition}
If $B, B'$ are linearly independent, then
$\mat{H}^{B,B'}_{\Lambda}$ is the matrix of
$H^{B,B'}_{\Lambda}$ in this basis $\{b_{1},\ldots,b_{r}\}$ of $\<B\>$ and the
dual basis of $B'$ in $\dual{ \<B'\>}$.
The {\em catalecticant} matrices of \cite{ia-book-1999} correspond to the case where
$B$ and $B'$ are respectively the set of monomials of degree $\le k$ and $\le
d-k$ ($k=0,\ldots,d$).

From the definition of the Hankel operators, we can deduce that a polynomial $p \in R$
belongs to the kernel of $\mat{H}_{\Lambda}$ if and only if
$p \star \Lambda = 0$, which in turn holds if and only if 
for all $q \in R$, $\Lambda(p q) = 0$.

\begin{proposition}
  \label{prop:kernel-is-ideal}
  Let $I_{\Lambda}$ be the kernel of ${H}_{\Lambda}$.
  Then, $I_{\Lambda}$ is an ideal of $R$.
\end{proposition}
\begin{proof}
  Let $p_1, p_2 \in I_{\Lambda}$. Then for all $q \in R$,
  $\Lambda( (p_1+p_2)q) = \Lambda(p_1 q) + \Lambda( p_2 q) = 0$.
  Thus, $p_1 + p_2 \in I_{\Lambda}$.
  If $p \in I_{\Lambda}$ and $p' \in R$,
  then for all $q \in R$, it holds $\Lambda (p p' q) = 0$. 
  Thus $p p' \in I_{\Lambda}$ 
  and $I_{\Lambda}$ is an ideal.
\end{proof}

Let $\AA = R / I_{\Lambda}$ be the quotient algebra of polynomials 
modulo the ideal $I_{\Lambda}$, which, as  Proposition~\ref{prop:kernel-is-ideal} 
states is the kernel of ${H_{\Lambda}}$. The rank of $H_{\Lambda}$  is the
dimension of $\AA$ as a $\kk$-vector space.

A quotient algebra $\Ac$ is Gorenstein if there exists a
non-degenerate bilinear form $Q$
on $\Ac$, such that for all polynomials $f, g, h \in \Ac$ it holds that
$Q( f, gh) = Q( fg, h)$ or equivalently if there exists $\Lambda \in
\dual{\Ac}$ such that $(f,g)\in \Ac \times \Ac \mapsto \Lambda(f\,g)$ is non-degenerate.
Equivalently, $\Ac$ is Gorenstein iff $\dual{\Ac}$ is a free $\Ac$-module
generated by one element $\Lambda \in \dual{\Ac}$:
$\dual{\Ac}=\Ac\star \Lambda$.  See e.g. \cite{EM08} for more details.
The set $\Ac\star\Lambda$ is also called the
inverse system generated by $\Lambda$ \cite{Mac16}.

\begin{proposition}
  The dual space $\dual{ \mathcal{A}}_{\Lambda}$ of $\mathcal{A}_{\Lambda}$,
  can be identified with the set $D = \Setbar{q \star \Lambda}{q \in R}$ and
$\AA$ is a Gorenstein algebra.
\end{proposition}
\begin{proof}
  Let $D =\{q \star \Lambda ; q \in R\}$ be the inverse system generated by
  $\Lambda$. By definition,
  \begin{displaymath}
    D^{\bot} = \Setbar{ p \in R}{ \forall q \in R, q \star \Lambda (p) =
      \Lambda (p q) = 0}. 
  \end{displaymath}
  Thus $D^{\bot} = I_{\Lambda}$, which is the ideal of the kernel of
  ${H}_{\Lambda}$ (Proposition~\ref{prop:kernel-is-ideal}).
  Since $\dual{ \mathcal{A}_{\Lambda}} = {I}_{\Lambda}^{\bot}$ 
  is the set of linear forms in $\dual{ R}$ which vanish on $I_{\Lambda}$, 
  we deduce that  $\dual{\AA}  =  I_{\Lambda}^{\bot} =
D^{\bot \bot} = D$.

As $p\star\Lambda=0$ implies $p\in I_{\Lambda}$ or $p\equiv 0$ in $\AA$, this
shows that  $\dual{\AA}$ is free rank $1$ $\AA$-module (generated by
$\Lambda$). Thus $\AA$ is Gorenstein.
\end{proof}

\begin{definition}
For any $B\subset R$, let $B^{+} = B\cup x_{1}B \cup \cdots \cup x_{n}B$ and
$\partial B= B^{+}-B$.
\end{definition}
 
\begin{proposition}\label{prop:basis:ideal}
  Assume that $\rank( {H}_{\Lambda}) = r < \infty$ 
  and let $B=\{b_1, \dots,  b_r\} \subset R$ such that 
  $\mat{H}_{\Lambda}^B$ is invertible. 
  Then $b_1, \dots, b_r$ is a basis of $\AA$. If $1 \in \<B\>$ 
the ideal $I_{\Lambda}$
is generated by $\ker H^{B^{+}}_{\Lambda}$.
\end{proposition}
\begin{proof}
  Let us first prove that $\Set{b_1, \ldots, b_r} \cap I_{\Lambda} =\{0\}$.
  Let $p \in \langle b_1, \ldots, b_r \rangle \cap I_{\Lambda}.$ 
  Then $p =  \sum_i{ p_i \,b_i}$ with $p_{i} \in \kk$ and $\Lambda (p \, b_j) = 0$.
  The second equation implies that $\mat{H}_{\Lambda}^B \cdot \bvec{p} = \bvec{0}$, 
  where $\bvec{p}=[p_1, \dots, p_r]^{t}\in \kk^{r}$. 
  Since $\mat{H}_{\Lambda}^B$ is invertible, this
  implies that $\bvec{p} = \bvec{0}$ and $p = 0$.

  As a consequence, we deduce that 
  $b_1 \star \Lambda, \dots, b_r \star \Lambda$ 
  are linearly independent elements of $\dual{R}$. 
  This is so, because otherwise there exists
  $\bvec{m} = [\mu_1, \ldots, \mu_r]^{\top} \neq \bvec{0}$, 
  such that 
  $\mu_1 (b_1 \star \Lambda) + \dots  + \mu_r (b_r \star \Lambda) = 
  (\mu_1 b_1 + \cdots + \mu_r b_r) \star \Lambda = 0$. 
  As $\Set{b_1, \ldots, b_r} \cap \mathsf{Kernel}( \mat{H}_{\Lambda}) = \set{ 0}$, this
  yields a contradiction.

  Consequently, $\set{b_1 \star \Lambda, \ldots, b_r \star \Lambda}$ span the image
  of ${H}_{\Lambda}$. For any $p \in R$, it holds that   
  $p \star \Lambda = \sum_{i = 1}^r{\mu_i  (b_i \star \Lambda)}$ 
  for some $\mu_1, \ldots, \mu_r \in \kk$. We
  deduce that $p - \sum_{i = 1}^r \mu_i b_i \in I_{\Lambda}$. This yields the
  decomposition 
  $R = B \oplus I_{\Lambda}$,
  and shows that
  $b_1, \ldots, b_r$ is a basis of $\AA$.

If $1\in \<B\>$, the ideal $I_{\Lambda}$ is generated by the relations $x_{j}
b_{k} - \sum_{i = 1}^r \mu_i^{j,k} b_i \in I_{\Lambda}$. These are precisely
in the kernel of $H_{\Lambda}^{B^{+}}$.
\end{proof}

\begin{proposition}
  \label{prop:form-decomp}
  If $\rank( {H}_{\Lambda}) = r < \infty$, 
  then $\AA$ is of dimension $r$ over $\kk$ and there exist
  $\zeta_1, \ldots, \zeta_d \in \kk^n$
  where $d \leq r$), and $p_i \in \kk[ \partial_1, \dots, \partial_n]$, such
  that
  \begin{equation}
    \Lambda = \sum_{i = 1}^d \mat{1}_{\zeta_i} \circ p_i ( \bvec{\partial})
    \label{eq:lambda}
  \end{equation}
  Moreover the multiplicity of $\zeta_{i}$ is the dimension of the vector space
  spanned the inverse system generated by $\mat{1}_{\zeta_i} \circ p_i ( \bvec{\partial})$.
\end{proposition}
\begin{proof}
  Since $\rank( \mat{H}_{\Lambda}) = r$, 
  the dimension of the vector space $\mathcal{A}_{\Lambda}$ is also $r$.
  Thus the number of zeros of the ideal $I_{\Lambda}$,
  say $\{\zeta_1, \ldots, \zeta_d \}$ is at most $r$, viz. $d \leq r$. 
  We can apply the structure Theorem \cite[Th.~7.34, p. 185]{EM08}
  in order to get the decomposition.
\end{proof}
In characteristic $0$, the inverse system of  $\mat{1}_{\zeta_i} \circ p_i (
\bvec{\partial})$ by $p_{i}$ is isomorphic to the vector space generated by
$p_{i}$ and its derivatives of any order with respect to the variables $\partial_{i}$.
In general characteristic, we replace the derivatives by the product by the
"inverse" of the variables \cite{mp-jcomplexity-2000}, \cite{EM08}.

\begin{definition} For $f\in S^{d}$, we call generalized decomposition of $\dual{f}$
a decomposition such that $\dual{f} = \sum_{i = 1}^d \mat{1}_{\zeta_i} \circ
p_i ( \bvec{\partial})$ where the sum for $i=1,\ldots,d$ of the dimensions of the vector spaces 
  spanned by the inverse system generated by $\mat{1}_{\zeta_i} \circ
p_i ( \bvec{\partial})$ is minimal. This minimal sum of dimensions is called the length of $f$.
\end{definition}
This definition extends the definition introduced in \cite{ia-book-1999} for
binary forms. The length of $\dual{f}$ is the rank of the corresponding
Hankel operator $H_{\Lambda}$.

\begin{theorem} \label{prop:decomp-rank}
Let $\Lambda \in \dual{R}$. 
  $\Lambda = \sum_{i = 1}^r \lambda_{i}\, \mat{1}_{\zeta_i}$
  with $\lambda_{i}\neq 0$ and $\zeta_{i}$ distinct points of $\kk^{n}$, iff
$\rank H_{\Lambda}=r$ and $I_{\Lambda}$ is a radical ideal.
\end{theorem}
\begin{proof}
If $\Lambda = \sum_{i = 1}^r \lambda_{i}\, \mat{1}_{\zeta_i}$, with
$\lambda_{i}\neq 0$ and $\zeta_{i}$ distinct points of $\kk^{n}$.
Let $\{e_{1},\ldots,e_{r}\}$ be a family of interpolation polynomials at these
points: $e_{i}(\zeta_{j})=1$ if $i=j$ and $0$ otherwise.
Let $I_{\zeta}$ be the ideal of polynomials which vanish at
$\zeta_{1},\ldots,\zeta_{r}$. It is a radical ideal. 
We have clearly $I_{\zeta}\subset I_{\Lambda}$. For any $p \in I_{\Lambda}$,
and $i=1,\ldots,r$, 
we have $p\star\Lambda(e_{i})= \Lambda(p\, e_{i}) = p(\zeta_{i})=0$, which
proves that $I_{\Lambda} = I_{\zeta}$ is a radical ideal.
As the quotient $\AA$ is generated by the interpolation polynomials
$e_{1},\ldots,e_{r}$, $H_{\Lambda}$ is of rank $r$.

Conversely, if $\rank H_{\Lambda}=r$, by Proposition \ref{prop:form-decomp}
$ \Lambda = \sum_{i = 1}^r \mat{1}_{\zeta_i} \circ p_i ( \bvec{\partial})$
with a polynomial of degree $0$, since the multiplicity of $\zeta_{i}$ is
$1$. This concludes the proof of the equivalence.
\end{proof}

In order to compute the zeroes of an ideal $I_{\Lambda}$ when we know a basis
of $\AA$, we exploit the properties of the operators of multiplication in
$\AA$:
$M_{a} : \AA \rightarrow \AA$, such that
$\forall b \in \AA, M_{a}(b)= a\, b$ and its transposed operator $M_{a}^{t} :
\dAA \rightarrow \dAA$, such that for 
$\forall  \gamma \in \dAA,  M_{a}^{\top}(\gamma) = a\star \gamma$.

The following proposition expresses a similar result, based on the properties
of the duality. 
\begin{proposition}
  \label{prop:Hl-mx}
  For any linear form $\Lambda \in \dual{R}$ such that 
  $\rank H_{\Lambda} < \infty$ and any $a \in \AA$, we have
  \begin{equation}
    H_{a \star \Lambda} = M_a^{t} \circ H_{\Lambda} \label{eq:multa}
  \end{equation} 
\end{proposition}
\begin{proof}
  By definition, $\forall p \in R, H_{a \star \Lambda} (p) = a\, p\star \Lambda
  = a\star  (p\star \Lambda) = M_{a}^{\top} \circ H_{\Lambda} (p)$.
\end{proof}
We have the following well-known theorem: 
\begin{theorem}   
  \label{th:stickelberger}
  Assume that $\AA$ is a finite dimensional vector space. Then 
  $\Lambda = \sum_{i = 1}^d \mat{1}_{\zeta_i} \circ p_i ( \bvec{\partial})$
  for $\zeta_{i}\in \kk^{n}$ and 
  $p_i ( \partial) \in \kk[ \partial_1, \dots, \partial_n]$ and 
  \begin{itemize}
  \item the eigenvalues of the operators ${M}_a$ and ${M}^{t}_a$,  
    are given by $\Set{ a( \zeta_1), \dots, a( \zeta_r)}$.
  \item the common eigenvectors of the operators $({M}_{x_i}^{t})_{1 \leq i \leq n}$ are
    (up to scalar) $\mat{1}_{\zeta_{i}}$.
  \end{itemize}
\end{theorem}
\begin{proof}
  \cite{CLO2,CLO,EM08}
\end{proof} 
Using the previous proposition, one can recover the points  $\zeta_{i}\in
\kk^{n}$ by eigenvector computation as follows.
Assume that $B\subset R$ with $|B|=\rank( H_{\Lambda})$, then equation \eqref{eq:multa}
and its transposition yield
$$ 
\mat{H}^{B}_{a \star \Lambda} = \mat{M}_{a}^{t} \mat{H}^{B}_{\Lambda} =
\mat{H}^{B}_{\Lambda} \, \mat{M}_{a},
$$  
where $\mat{M}_{a}$ is the matrix of multiplication by $a$ in the basis $B$
of $\AA$. By Theorem \ref{th:stickelberger}, the common solutions of the generalized
eigenvalue problem 
\begin{equation}\label{eq:geiv}
  (\mat{H}_{a \star \Lambda} - \lambda \, \mat{H}_{\Lambda}) \bvec{v} = \mat{O}
\end{equation}
for all $a\in R$, 
yield the common eigenvectors $\mat{H}^{B}_{\Lambda} \bvec{v}$ of
$\mat{M}_a^{t}$, that is the evaluation $\mat{1}_{\zeta_{i}}$ at the roots. 
Therefore, these common eigenvectors $\mat{H}^{B}_{\Lambda} \bvec{v}$ are up to a
scalar, the vectors $[b_{1}(\zeta_{i}),\ldots,b_{r}(\zeta_{i})]$
$(i=1,\ldots,r)$.
Notice that it is sufficient to compute the common eigenvectors of 
 $(\mat{H}_{x_{i} \star \Lambda} , \mat{H}_{\Lambda})$ for $i=1,\ldots, n$

If $\Lambda = \sum_{i = 1}^d \lambda_{i} \mat{1}_{\zeta_i}$ $(\lambda_{i}\neq
0)$, then the roots are simple, and one eigenvector computation is enough:
for any $a\in R$, $\mat{M}_{a}$ is diagonalizable and the
generalized eigenvectors $\mat{H}^{B}_{\Lambda} \bvec{v}$ are, up to a scalar,
the evaluation $\mat{1}_{\zeta_{i}}$ at the roots.

\section{Truncated Hankel operators}\label{sec:trunc:hankel}\label{sec:5}
Coming back to our problem of symmetric tensor decomposition, $f
=\sum_{|\alpha|\leq d} c_{\alpha}{d \choose \alpha} \x^{\alpha} \in R_{d}$ admits
an affine decomposition of rank $r$, iff $\Lambda(\x^{\alpha})=c_{\alpha}$ for
all $|\alpha|\leq d$ where 
$$ 
\Lambda = \sum_{i=1}^{r} \lambda_{i}\, \mat{1}_{\zeta_{i}},
$$ 
for some distinct $\zeta_{1},\ldots, \zeta_{r} \in \kk^{n}$ and
some $\lambda_{i}\in \kk-\{0\}$. 

Then, by theorem \ref{prop:decomp-rank},
 $H_{\Lambda}$ is of rank $r$ and $I_{\Lambda}$ is radical. 

Conversely, given $H_{\Lambda}$ of rank $r$ with  $I_{\Lambda}$ radical
which coincides up to degree $d$ with $\Lambda_{d}$, by proposition
\ref{prop:form-decomp}, 
$\Lambda= \sum_{i=1}^{r} \lambda_{i}\, \mat{1}_{\zeta_{i}}$ and 
$f$ can be decomposed as a sum of $r$ $d^{th}$-powers of linear forms.

The problem of decomposition of $f$ can thus be reformulated as follows:
\begin{quote}\em 
  Given $\dual{f}\in \dual{R_{d}}$ find the smallest $r$ such that there exists 
  $\Lambda \in  \dual{R}$ which extends $\dual{f}$  with $H_{\Lambda}$ of rank $r$
  and $I_{\Lambda}$ a radical ideal.
\end{quote}
In this section, we are going to characterize under which conditions
$\dual{f}$ can be extended to $\Lambda \in \dRR$ with $H_{\Lambda}$ is of
rank $r$.

We need the following technical property on the bases of $\AA$, that we will consider:
\begin{definition}
Let $B$ be a subset of monomials in $R$. We say that $B$ is
connected to $1$ if $\forall m\in B$ either $m=1$ or there exists $i\in
[1,n]$ and $m'\in B$ such that $m=x_{i}\, m'$.
\end{definition}

Let $B \subset  R_{d}$ be a set of monomials of degree $\le d$, connected to $1$.
We consider the formal Hankel matrix 
$$ 
\form{H}^{B}_{\Lambda} = (h_{\alpha+\beta})_{\alpha, \beta\in B},
$$
with $h_{\alpha}=\dual{f}(\x^{\alpha})=c_{\alpha}$ if 
$|\alpha|\le d$ and otherwise $h_{\alpha}$ is a variable.
The set of all these new variables is denoted $\hb$.

Suppose that $\form{H}^{B}_{\Lambda}$ is invertible in $\kk(\hb)$, then we
define the formal multiplication operators 
$$
\form{M}_{i}^{B}(\hb) := (\form{H}^{B}_{\Lambda})^{-1} \form{H}^{B}_{x_{i}\star
  \Lambda}.
$$
The following result characterizes the cases where 
$\kk[\x] = B \oplus I_{\Lambda}$:
\begin{theorem}\label{th:commute}
Let $B=\{\x^{\beta_{1}},\ldots,\x^{\beta_{r}}\}$ be a set of monomials of
degree $\le d$, connected to $1$ and $\Lambda$ be a linear form in
$\dual{\<B\cdot B^{+}\>}_{d}$. Let $\Lambda(\hb)$ be the linear form of $\dual{\<B\cdot B^{+}\>}$ defined by $\Lambda(\hb)(\x^{\alpha})=\Lambda(\x^{\alpha})$ if $|\alpha|\le d$ and $h_{\alpha} \in \kk$ otherwise. Then, $\Lambda(\hb)$ admits an extension $\tilde{\Lambda} \in \dRR$ such that $H_{\tilde{\Lambda}}$ is of rank $r$ with $B$ a basis of $A_{\tilde{\Lambda}}$ iff 
\begin{equation}
  \form{M}_{i}^{B}(\hb) \circ \form{M}_{j}^{B}(\hb)
  -\form{M}_{j}^{B}(\hb) \circ \form{M}_{i}^{B}(\hb)=0 \ \ (1 \le i<j\le n)
\end{equation}
and $\det(\form{H}^{B}_{\Lambda})(\hb) \neq 0$. Moreover, such a $\tilde{\Lambda}$ is unique.
\end{theorem}
\begin{proof} 
If there exists $\tilde{\Lambda} \in \dRR$ which extends
$\Lambda(\hb)$, with $H_{\tilde{\Lambda}}$ of rank $r$ then the
tables of multiplications by the variables $x_{i}$ are 
$M_{i}=(\form{H}^{B}_{\Lambda})^{-1} \form{H}^{B}_{x_{i}\star \Lambda}$
(proposition \ref{prop:Hl-mx}) and they commute.

Conversely suppose that these matrices commute. 
Then by \cite{mourrain-aaecc-1999}, we have 
$\kk[\x] = \<B\> \oplus (K)$, where
$K$ is the vector space generated by the border relations $x_{i} m -
\form{M}_{i} (m)$ for $m\in B$ and $i=1,\ldots, n$.  Let $\pi_{B}$ be the
projection of $R$ on $\<B\>$ along $(K)$.

We define $\tilde{\Lambda} \in \dRR$ as follows:  
$\forall p\in \RR, \tilde{\Lambda}(p)=\Lambda(p(\form{M})(1))$ where
$p(\form{M})$ is the operator obtained by substitution of the variables
$x_{i}$ by the commuting operators $\form{M}_{i}$.
Notice that $p(\form{M})$ is also the operator of multiplication by $p$
modulo $(K)$. 

By construction, $(K) \subset \ker H_{\tilde{\Lambda}}$ and $B$ is a
generating set of $\Ac_{\tilde{\Lambda}}$.

Let us prove by induction on the degree of $b\in B$ that for all $b'\in B$, we
have $\Lambda(b\, b')= \Lambda(b(M)\, (b'))$. The property is true for $b=1$.
As $B$ is connected to $1$, if $b\neq 1$, then $b = x_{i}\, b''$ for some
variable $x_{i}$ and some element $b''\in B$ of degree smaller than $b$. 
By construction of the operators $M_{i}$, we have $\Lambda(x_{i} b''\, b')=
\Lambda(b'' M_{i}(b'))$. 
By induction hypothesis, we deduce that 
$\Lambda(b \, b')= \Lambda(b''(M) \circ M_{i}(b'))= \Lambda(b(M)(b'))$.
As $b'=b'(M)(1)$ for all $b'\in B$ (the multiplication of $1$ by $b$
is represented by $b\in \<B\>$ modulo $(K)$), we deduce that  
$$ 
\Lambda(b \, b')= \Lambda(b''(M) \circ M_{i}(b'))= \Lambda(b(M)\circ b'(M)(1))=
\Lambda((b\, b')(M)(1))= \tilde{\Lambda}(b\,b').
$$
This shows that $\Lambda=\tilde{\Lambda}$ on $B \cdot B$. As
$\det(\form{H}^{B}_{\Lambda})\neq 0$, we deduce that $B$ is a basis of
$\Ac_{\tilde{\Lambda}}$ and that $H_{\tilde{\Lambda}}$ is of rank $r$.

If there exists another $\Lambda' \in \dRR$ which extends $\Lambda(\hb) \in
\dual{\<B\cdot B^{+}\>}$ with $\rank H_{\Lambda'}=r$, by proposition
\ref{prop:basis:ideal}, $\ker H_{\Lambda'}$ is generated by
$\ker H_{\Lambda'}^{B\cdot B^{+}}$ and thus coincides with $\ker
H_{\tilde{\Lambda}}$. As $\Lambda'$ coincides with $\tilde{\Lambda}$ on $B$,
the two elements of $\dRR$ must be equal.  This ends the proof of the
theorem.
\end{proof}
 
The degree of these commutation relations is at most $2$ in the coefficients
of the multiplications matrices $\form{M}_{i}$. A direct computation yields the following, for 
$m\in B$:
\begin{itemize}
\item If $x_{i}, m \in B, x_{j}\,m \in B$ then $(\form{M}_{i}^{B} \circ \form{M}_{j}^{B}
  -\form{M}_{j}^{B} \circ \form{M}_{i}^{B})(m)\equiv 0$ in $\kk(\hb)$.
\item If $x_{i} m \in B$, $x_{j}\,m \not\in B$ then 
  $(\form{M}_{i}^{B} \circ \form{M}_{j}^{B} - \form{M}_{j}^{B} \circ
  \form{M}_{i}^{B})(m)$ is of degree $1$ in the coefficients of
  $\form{M}_{i},\form{M}_{j}$.
\item If $x_{i} m \not\in B$, $x_{j}\,m \not\in B$ then 
  $(\form{M}_{i}^{B} \circ \form{M}_{j}^{B} - \form{M}_{j}^{B}
  \circ \form{M}_{i}^{B})(m)$ is of degree $2$ in the coefficients of
  $\form{M}_{i},\form{M}_{j}$.
\end{itemize}

We are going to give an equivalent characterization of the extension 
property, based on rank conditions.

\begin{theorem}\label{thm:rank}
Let $B=\{\x^{\beta_{1}},\ldots,\x^{\beta_{r}}\}$ be a set of monomials of
degree $\le d$, connected to $1$. Then, the linear form $\dual{f}\in
\dual{S}_{d}$ admits an extension $\Lambda \in \dRR$ such that
$H_{\Lambda}$ is of rank $r$ with $B$ a basis of $\AA$ iff there exists an $\hb$ such that all
$(r+1)\times(r+1)$ minors of $\form{H}^{B^{+}}_{\Lambda}(\hb)$ vanish and
$\det(\form{H}^{B}_{\Lambda})(\hb) \neq 0$.
\end{theorem}
\begin{proof}
Clearly, if there exists $\Lambda \in \dRR$ which extends
$\dual{f}\in \dual{S}_{d}$ with $H_{\Lambda}$ of rank $r$, then all
$(r+1)\times(r+1)$ minors of $\form{H}^{B^{+}}_{\Lambda}(\hb)$ vanish.

Conversely, if $\form{H}^{B^{+}}_{\Lambda}(\hb)$ and
$\form{H}^{B}_{\Lambda}(\hb)$ are of rank $r$, by \cite[Theorem 1.4]{ML08} there
exists a unique $\tilde{\Lambda}\in R^{*}$ such that $H_{\Lambda}$ is of rank
$r$, and which coincides with $\Lambda$ on $\<B^{+}\cdot B^{+}\>$.
\end{proof}

\begin{proposition}\label{prop:dec}
Let $B=\{\x^{\beta_{1}},\ldots,\x^{\beta_{r}}\}$ be a set of monomials of
degree $\le d$, connected to $1$. Then, the linear form $\dual{f}\in
\dual{S}_{d}$ admits an extension $\Lambda \in \dRR$ such that
$H_{\Lambda}$ is of rank $r$ with $B$ a basis of $\AA$ iff 
\begin{equation}\label{dec:H}
\mat{H}^{B^{+}}_{\Lambda} 
  =
  \left(
    \begin{array}{cc}
      \mat{H} & \mat{G} \\
      \mat{G}^{t} & \mat{J}  
    \end{array}
  \right),
\end{equation}
with $\mat{H}= \mat{H}^{B}_{\Lambda}$ and 
\begin{equation}\label{eq:W}
\mat{G}=\mat{H}\, \mat{W},
\mat{J}=\mat{W}^{t}\,\mat{H}\,\mat{W}. 
\end{equation} 
for some matrix $\mat{W}\in \kk^{B \times \partial B}$.
\end{proposition}
\begin{proof} 
According to theorem \ref{thm:rank}, $\dual{f}\in \dual{S}_{d}$ admits a (unique)
extension $\Lambda \in \dRR$ such that 
$H_{\Lambda}$ is of rank $r$ with $B$ a basis of $\AA$, iff 
$H^{B^{+}}$ is of rank $r$.
Let us decompose $H^{B^{+}}$ as \eqref{dec:H} with $\mat{H}=
\mat{H}^{B}_{\Lambda}$. 

If we have $\mat{G}=\mat{H}\, \mat{W},
  \mat{J}=\mat{W}^{t}\,\mat{H}\,\mat{W}$, then 
$$
  \left(
    \begin{array}{cc}
      \mat{H} & \mat{H}\, \mat{W} \\
      \mat{W}^{t} \mat{H} & \mat{W}^{t}\,\mat{H}\,\mat{W}
    \end{array}
  \right)
$$
is clearly of rank  $\le \rank \mat{H}$.

Conversely, suppose that $\mat{H}^{B^{+}}_{\Lambda} = \rank \mat{H}$. 
This implies that the image of $\mat{G}$ is in the image of $\mat{H}$. Thus,
there exists  $\mat{W}\in \kk^{B \times \partial B}$ such that
$\mat{G}=\mat{H}\, \mat{W}$. 
Without loss of generality, we can assume that the $r$ first columns of
$\mat{H}$ ($r=\rank \mat{H}$) are linearly independent. Assume that 
we choose $\mat{W}$ such that the $i^{\mathrm{th}}$ column of
$\mat{G}$ is the linear combination of the $r$ first columns with
coefficients corresponding to the $i$ column $\mat{W}_{i}$ of $\mat{W}$. 
As $\rank \mat{H}^{B^{+}}_{\Lambda}=r$ the same relation holds for the whole column
of this matrix. Thus we have $\mat{J}= \mat{G}^{t}\, \mat{W}=
 \mat{W}^{t}\,\mat{H}\,\mat{W}$. 
\end{proof}
Notice that if $\mat{H}$ is invertible, $\mat{W}$ is uniquely determined.
In this case, we easily check that 
$\ker \mat{H}^{B^{+}}_{\Lambda} = \left(
\begin{array}{c}
\mat{W} \\
-\mat{I}
\end{array}\right)$.

This leads to the following system in the variables $\bvec{h}$ and the coefficients
$\bvec{w}$ of matrix $\mat{W}$. It characterizes the linear forms $\dual{f}\in
\dual{S}_{d}$ that admit an extension $\Lambda \in \dRR$ such that
$H_{\Lambda}$ is of rank $r$ with $B$ a basis of $\AA$.
\begin{equation} \label{syst:GH}
\form{H}_{\Lambda}^{B,\partial B}(\bvec{h}) - \form{H}_{\Lambda}^{B}(\bvec{h})\,
\mat{W}(\bvec{w}) =0, \ \ 
\form{H}_{\Lambda}^{\partial B, \partial B}(\bvec{h}) - 
\mat{W}^{t}(\bvec{w})\,\form{H}_{\Lambda}^{B}(\bvec{h})\,\mat{W}(\bvec{w}) =0
\end{equation}
with $\det(\form{H}_{\Lambda}^{B}(\bvec{h}))\neq 0$.

The matrix $\form{H}^{B^{+}}_{\Lambda}$ is a quasi-Hankel matrix
\cite{mp-jcomplexity-2000}, whose structure is imposed by equality (linear)
constraints on its entries.
If $\mat{H}$ is known (ie. $B\times B \subset R_{d}$, the number of
independent parameters in $\form{H}_{\Lambda}^{B,B^{+}}(\bvec{h})$ or in
$\mat{W}$ is the number  of monomials in $B \times \partial B - R_{d}$.   
By Proposition \ref{prop:dec}, the rank condition is equivalent to the quadratic
relations $\mat{J}- \mat{W}^{t} \mat{H}^{t}\, \mat{W}=0$ in these unknowns. 

If $\mat{H}$ is not completely known, the number of parameters in $\mat{H}$ 
is the number of monomials in $B\times B - R_{d}$. The number of independent parameters in
$\form{H}_{\Lambda}^{B,\partial B}(\bvec{h})$ or in $\mat{W}$ is then $B\times
\partial B - R_{d}$. 

The system \eqref{syst:GH} is composed of linear equations deduced from 
 quasi-Hankel structure, quadratic relations  for the entries in $B\times
\partial B$
and cubic relations  for the entries in $B\times \partial B$ in the unknown
parameters $\bvec{h}$ and $\bvec{w}$.

We are going to use explicitly these characterizations in the new algorithm
we propose for minimal tensor decomposition.

%
%
%
 
\section{Symmetric tensor decomposition algorithm}\label{sec:6}
The algorithm that we will present for decomposing a symmetric tensor as sum
of rank 1 symmetric tensors generalizes the algorithm of Sylvester
\cite{sylv-cr-1886}, devised for dimension 2 tensors, see also
\cite{cm-binary-arxiv-2001}. 


Consider the homogeneous polynomial $f(\bvec{x})$ in \eqref{eq:homo-poly}
that we want to decompose. We may assume without loss of generality, that for
at least one variable, say $x_0$, all its coefficients in the decomposition
are non zero, i.e. $k_{i,0} \not= 0$, for $1 \leq i \leq r$.  We dehomogenize
$f$ with respect to this variable and we denote this polynomial by $f^{a} :=
f(1,x_{1},\ldots,x_{n})$.
We want to decompose the polynomial $f^{a}(\bvec{x}) \in R_d$  as a sum of powers of linear forms, i.e.
$$ 
  f(\bvec{x}) = \sum_{i=1}^{r}{ \lambda_{i}\, ( 1  + k_{i,1} x_1 + \cdots + k_{i, n} x_n)^d} 
 =   \sum_{i=1}^{r}{
\lambda_{i}\, \bvec{k}_i(\bvec{x})^d}
$$
Equivalently, we want to decompose its corresponding dual element $\dual{f} \in
\dual{R_{d}}$ 
as a linear combination of evaluations  over the distinct points $ \bvec{k}_{i} := (k_{i,1},
\cdots, k_{i, n})$: 
$$ 
\dual{f} = \sum_{i=1}^{r}{ \lambda_{i}\, \mat{1}_{\bvec{k}_i}}
$$ (we refer the reader to the end of Section~\ref{sec:duality}).

Assume that we know the value of $r$. As we have seen previously, knowing 
the value of $\Lambda$ on polynomials of degree high enough, allows us to compute
the table of multiplications modulo the kernel of $\mat{H}_{\Lambda}$. 
By  Theorem ~\ref{th:stickelberger}, 
solving the generalized eigenvector problem
$(\mat{H}_{x_1 \star \Lambda} - \lambda \, \mat{H}_{\Lambda}) \bvec{v} = \mat{O}$,
we will recover the points of evaluation $\bvec{k}_{i}$. By solving a linear
system, we will then deduce the value of $\lambda_{i},\ldots,\lambda_{r}$. Thus,
the goal of the following algorithm is to extend $\dual{f}$ on a large enough
set of polynomials, in order to be able to run this eigenvalue computation.

\begin{algorithm}\caption{\textsc{Symmetric tensor decomposition}}
\KwIn{A homogeneous polynomial $f(x_{0},x_{1},\ldots,x_{n})$ of degree $d$.}
\KwOut{A decomposition of $f$ as $f=\sum_{i=1}^{r} \lambda_{i}\, \bvec{k}_i(\bvec{x})^d$ with $r$ minimal.}
  \begin{itemize}
  \item[--] Compute the coefficients of $\dual{f}$: 
$c_{\alpha}=a_{\alpha}\, {d \choose \alpha}^{-1}$, for $|\alpha|\leq d$;
  \item[--] $r:=1$;

  \item[--] \textbf{Repeat}
    \begin{enumerate}
    \item Compute a set $B$ of monomials of degree $\le d$ connected to 1
with $|B|=r$; 

    \item Find parameters $\hb$ s.t. 
$\det (\mat{H}^{B}_{\Lambda}) \neq 0$
and the      operators $\mat{M}_{i}=\mat{H}^{B}_{x_{i}\Lambda}(\mat{H}^{B}_{\Lambda})^{-1}$
      commute.
    \item If there is no solution, restart the loop with $r:=r+1$.

    \item Else compute the $n\times r$ eigenvalues $\zeta_{i,j}$ and the
eigenvectors $\mathbf{v}_{j}$ s.t. $\mat{M}_{i}\mathbf{v}_{j}
      =\zeta_{i,j}\mathbf{v}_{j}$, $i=1,\ldots,n$, $j=1,\ldots,r$.

    \end{enumerate}
    \textbf{until} the eigenvalues are simple.

  \item[--] Solve the linear system in $(\nu_j)_{j=1,\ldots,k}$:
    $\Lambda=\sum_{j=1}^{r}\nu_{j}\mathbf{1}_{\zeta_{j}}$ 
where $\zeta_{j}\in
    \kk^{n}$ are the eigenvectors found \newline 
in step 4.
  \end{itemize}
\end{algorithm}

The critical part in this algorithm is the completion of step 2.
Instead of the commutation relations, one can use the result of Proposition 
\ref{prop:dec}. 

\subsection{First Example}
\label{sec:ex-1}
The example that follows will make the steps of the algorithm clearer.

\begin{enumerate}
\item Convert the symmetric tensor to the corresponding homogeneous polynomial.
  
  \EXAMPLE{
    Assume that a tensor of dimension 3 and order 5, 
    or equivalently a 3-way array of dimension 5,
    corresponds to the following homogeneous polynomial \\
    
    $f = -1549440\,x_0x_1{x_2}^{3}+2417040\,x_0{x_1}^{2}{x_2}^{2}
    +166320\,{x_0}^{2}x_1{x_2}^{2}-829440\,x_0{x_1}^{3}x_2
    -5760\,{x_0}^{3}x_1x_2-222480\,{x_0}^{2}{x_1}^{2}x_2
    +38\,{x_0}^{5}-497664\,{x_1}^{5}-1107804\,{x_2}^{5}
    -120\,{x_0}^{4}x_1+180\,{x_0}^{4}x_{{2}}+12720\,{x_0}^{3}{x_1}^{2}
    +8220\,{x_0}^{3}{x_2}^{2}-34560\,{x_0}^{2}{x_1}^{3}-59160\,{x_0}^{2}{x_2}^{3}
    +831840\,x_0{x_1}^{4}+442590\,x_0{x_2}^{4}-5591520\,{x_1}^{4}x_2
    +7983360\,{x_1}^{3}{x_2}^{2}-9653040\,{x_1}^{2}{x_2}^{3}+5116680\,x_1{x_2}^{4}$.

    The minimum decomposition of the polynomial as a sum of powers of linear forms is 
    \begin{displaymath}
      ( x_0+2\,x_1+3\,x_2 )^5 + ( x_0-2\,x_1+3\,x_2 )^5 +
      \frac{1}{3} ( x_0 -12\,x_1-3\,x_2 )^5 + \frac{1}{5} ( x_0+12\,x_1-13\,x_2 )^5,
    \end{displaymath}
    that is, the corresponding tensor is of rank 4.
  }

\item Compute the actual number of variables needed.
  
  For algorithms computing the so-called number of {\em essential} variables, the
  reader may refer to the work of Oldenburger 
  \cite{olden-annals-1934} or Carlini \cite{carlini-aggm-2005}.
  
  \EXAMPLE{
    In our example the number of essential variable is 3, so we have nothing to do.
  }

\item Compute the matrix of the quotient algebra.

  We form a ${n+d-1 \choose d} \times {n+d-1 \choose d}$ matrix,
  the rows and the columns of which correspond to the coefficients
  of the polynomial in the dual base.
  The map for this is 
  \begin{displaymath}
    a_{j_0 \, j_1 \, \dots \, j_n} \mapsto  c_{j_0 \, j_1 \, \dots \, j_n} \,
:= a_{j_0 \, j_1 \, \dots \, j_n} \, 
    {d \choose j_{0}, \ldots, j_{n}}^{-1},
  \end{displaymath}
  where  $a_{d_0 \, d_1 \, \dots \, d_n}$ is the coefficient of the monomial
  $x_{0}^{j_{0}}\cdots x_{n}^{j_{n}}$ in $f$.
  Recall that, since the polynomial is homogeneous,
  $\sum_{i=1}^{n}{j_i} = d$.
  
  This matrix is called quasi-Hankel \cite{mp-jcomplexity-2000} or Catalecticant
  \cite{ia-book-1999}. 
  
  \EXAMPLE{
    Part of the corresponding matrix follows. The whole matrix is $21 \times 21$.
    We show only the $10 \times 10$  principal minor.
    \begin{displaymath}
      \scriptsize
      \left[ 
        \begin {array}{c|rrrrrrrrrr} 
          & 1 & x_1 & x_2 & x_1^2 & x_1 x_2 & x_2^2 & x_1^3 & x_1^2x_2 & x_1 x_2^2 &
          x_2^3\\ \hline

          1 & 38&-24&36&1272&-288&822&-3456&-7416&5544&-5916\\
          x_1 & -24&1272&-288&-3456&-7416&5544&166368&-41472&80568&-77472\\
          x_2 & 36&-288&822&-7416&5544&-5916&-41472&80568&-77472&88518\\
          x_1^2 & 1272&-3456&-7416&166368&-41472&80568&-497664&-1118304&798336&-965304\\
          x_1 x_2 & -288&-7416&5544&-41472&80568&-77472&-1118304&798336&-965304&1023336\\
          x_2^2 & 822&5544&-5916&80568&-77472&88518&798336&-965304&1023336&-1107804\\
          x_1^3 & -3456&166368&-41472&-497664&-1118304&798336&h_{{6,0,0}}&h_{{5,1,0}}&h_{{4,2,0}}&h_{{3,3,0}}\\
          x_1^2 x_2 & -7416&-41472&80568&-1118304&798336&-965304&h_{{5,1,0}}&h_{{4,2,0}}&h_{{3,3,0}}&h_{{2,4,0}}\\
          x_1 x_2^2 & 5544&80568&-77472&798336&-965304&1023336&h_{{4,2,0}}&h_{{3,3,0}}&h_{{2,4,0}}&h_{{1,5,0}}\\
          x_2^3 & -5916&-77472&88518&-965304&1023336&-1107804&h_{{3,3,0}}&h_{{2,4,0}}&h_{{1,5,0}}&h_{{0,6,0}}
        \end {array}
      \right] 
    \end{displaymath}

    Notice that we do not know the elements in some positions of the matrix. 
    In general we do not know the elements that correspond to monomials with (total) 
    degree higher than 5.
  }
  
\item Extract a principal minor of full rank.
  
  We should re-arrange the rows and the columns of the matrix so that there
  is a principal minor of full rank, $R$. We call this minor $\Delta_0$.  In
  order to do that we try to put the matrix in row echelon form, using
  elementary row and column operations.
  
  \EXAMPLE{
    In our example the $4 \times 4$ principal minor is of full rank, 
    so there is no need for re-arranging the matrix.
    The matrix $\Delta_0$ is 
    \begin{displaymath}
      \Delta_0 = 
      \left[ 
        \begin {array}{rrrr} 
          38&-24&36&1272\\\noalign{\medskip}
          -24&1272&-288&-3456\\\noalign{\medskip}
          36&-288&822&-7416\\\noalign{\medskip}
          1272&-3456&-7416&166368
        \end {array} 
      \right] 
    \end{displaymath}
    Notice that the columns of the matrix correspond to the monomials
    $\{ 1, x_1, x_2, x_1^2\}$.
  }

\item We compute the ``shifted'' matrix $\Delta_1 = x_1 \Delta_0$.

  The columns of $\Delta_0$ correspond to set of some monomials, 
  say $\set{ \bvec{x}^{\boldsymbol{\alpha}}}$ where $\boldsymbol{\alpha} \subset \NN^{n}$.
  The columns of $\Delta_1$ correspond to the set of monomials
  $\set{x_1 \, \bvec{x}^{\boldsymbol{\alpha}}}$.
  
  \EXAMPLE{
    The shifted matrix $\Delta_1$ is 
    \begin{displaymath}
      \Delta_1 = 
      \left[ 
        \begin {array}{rrrr} 
          -24&1272&-288&-3456\\\noalign{\medskip}
          1272&-3456&-7416&166368\\\noalign{\medskip}
          -288&-7416&5544&-41472\\\noalign{\medskip}
          -3456&166368&-41472&-497664
        \end {array} 
      \right] 
    \end{displaymath}
    Notice that the columns  correspond to the monomials
    $\{ x_1, x_1^2, x_1 x_2, x_1^3\}$, 
    which are just the corresponding monomials of the columns of $\Delta_0$,
    i.e. $\{ 1, x_1, x_2, x_1^2\}$, multiplied by $x_1$.
  }

  We assume for the moment that all the elements of the matrices $\Delta_0$ and
  $\Delta_1$ are known. If this is not the case, then we can compute the unknown entries
  of the matrix, using either necessary and sufficient conditions of the quotient algebra,
  e.g. it holds that $\mat{M}_{x_i} \mat{M}_{x_j} - \mat{M}_{x_j} \mat{M}_{x_i} = \mat{O}$ \cite{mourrain-aaecc-1999}
  for any $i, j \in \{1, \dots, n\}$.
  There are other algorithms to extend a moment matrix,
  e.g. \cite{laurent-ima-2008,laurent-ams-2005,cf-hjm-1991}.

\item We solve the equation $(\Delta_1 - \lambda \Delta_0) X = 0$.

  We solve the generalized eigenvalue/eigenvector problem using one of the well-known
  techniques \cite{gvl-book-1996}.
  We normalize the elements of the eigenvectors so that the first element is 1, 
  and we read the solutions from the coordinates of the (normalized) eigenvectors.
  
  \EXAMPLE{
    The normalized eigenvectors of the generalized eigenvalue problem are
    \begin{displaymath}
      \left[ \begin {array}{r} 
          1 \\ -12\\ -3\\  144
        \end {array} \right] ,
      \left[ \begin {array}{r} 
          1\\ 12\\ -13\\ 144
        \end {array} \right] , 
      \left[ \begin {array}{r} 
          1\\ -2\\ 3\\ 4
        \end {array} \right] , 
      \left[ \begin {array}{r} 
          1\\ 2\\ 3\\ 4
        \end {array}
      \right]
    \end{displaymath}

    The coordinates of the eigenvectors correspond to the elements
    $\{ 1, x_1, x_2, x_1^2\}$.
    Thus, we can recover the coefficients of $x_1$ and $x_2$ 
    in the decomposition from  coordinates of the eigenvectors.

    Recall that the coefficients of $x_0$ are considered to be one.
    Thus, The polynomial admits a decomposition
    \begin{displaymath}
      f =  \ell_1 ( x_0 -12 x_1  -3 x_2)^5 +
      \ell_2 ( x_0 + 12 x_1  -13 x_2)^5 +
      \ell_3 ( x_0  -2 x_1  +3 x_2)^5 +
      \ell_4 ( x_0 + 2 x_1  + 3 x_2)^5
    \end{displaymath}
    It remains to compute $\ell_i$'s. We can do this easily by solving an
    over-determined linear system, which we know that always has a solution, 
    since the decomposition exists.
    Doing that, we deduce that 
    $\ell_1 = 3$, $\ell_2 = 15$, $\ell_3 = 15$ and $\ell_4 = 5$.
  }
\end{enumerate}

\subsection{Second Example}
\label{sec:ex-2}

One of the assumptions that the previous example fulfills is that all the entries of the
matrices needed for the computations are known. 
However, this is not always the case as the following example shows.

\begin{enumerate}
\item Convert the symmetric tensor to the corresponding homogeneous polynomial. 
  \EXAMPLE{
    Consider a tensor of dimension 3 and order 4, 
    that corresponds to the following homogeneous polynomial
    $$f = 79\,x_0 x_1^3 + 56\,x_0^2 x_2^2 +49\,x_1^2 x_2^2 +4\,x_0 x_1 x_2^2 +57\,x_0^3x_1,$$
    the rank of which is 6.
  }

\item Compute the actual number of variables needed.
  \EXAMPLE{
    In our example the number of essential variables is 3, so we have nothing to do.
  }

\item Compute the matrix of the quotient algebra.
  \EXAMPLE{
    
    The matrix is $15 \times 15$.
    \begin{displaymath}
      \scriptsize     \left[ 
        \begin {array}{c|ccccccccccccccc} 
          &1&x_1&x_2&x_1^2&x_1x_2&x_2^2&x_1^3&x_1^2x_2&x_1x_2^2&x_2^3&x_1^4&x_1^3x_2&x_1^2x_2^2&x_1x_2^3& x_2^4 \\ \hline
          1&0&\frac{57}{4}&0&0&0&\frac{28}{3}&\frac{79}{4}&0&\frac{1}{3}&0&0&0&\frac{49}{6}&0&0 \\
          x_1&\frac{57}{4}&0&0&\frac{79}{4}&0&\frac{1}{3}&0&0&\frac{49}{6}&0&h_{500}&h_{410}&h_{320}&h_{230}&h_{140}\\
          x_2&0&0&\frac{28}{3}&0&\frac{1}{3}&0&0&\frac{49}{6}&0&0&h_{410}&h_{320}&h_{230}&h_{140}&h_{050}\\
          x_1^2&0&\frac{79}{4}&0&0&0&\frac{49}{6}&h_{500}&h_{410}&h_{320}&h_{230}&h_{600}&h_{510}&h_{420}&h_{330}&h_{240}\\
          x_1x_2&0&0&\frac{1}{3}&0&\frac{49}{6}&0&h_{410}&h_{320}&h_{230}&h_{140}&h_{510}&h_{420}&h_{330}&h_{240}&h_{150}\\
          x_2^2&\frac{28}{3}&\frac{1}{3}&0&\frac{49}{6}&0&0&h_{320}&h_{230}&h_{140}&h_{050}&h_{420}&h_{330}&h_{240}&h_{150}&h_{060}\\
          x_1^3&\frac{79}{4}&0&0&h_{500}&h_{410}&h_{320}&h_{600}&h_{510}&h_{420}&h_{330}&h_{700}&h_{610}&h_{520}&h_{430}&h_{340}\\
          x_1^2x_2&0&0&\frac{49}{6}&h_{410}&h_{320}&h_{230}&h_{510}&h_{420}&h_{330}&h_{240}&h_{610}&h_{520}&h_{430}&h_{340}&h_{250}\\
          x_1x_2^2&\frac{1}{3}&\frac{49}{6}&0&h_{320}&h_{230}&h_{140}&h_{420}&h_{330}&h_{240}&h_{150}&h_{520}&h_{430}&h_{340}&h_{250}&h_{160}\\
          x_2^3&0&0&0&h_{230}&h_{140}&h_{050}&h_{330}&h_{240}&h_{150}&h_{060}&h_{430}&h_{340}&h_{250}&h_{160}&h_{070}\\
          x_1^4&0&h_{500}&h_{410}&h_{600}&h_{510}&h_{420}&h_{700}&h_{610}&h_{520}&h_{430}&h_{800}&h_{710}&h_{620}&h_{530}&h_{440}\\
          x_1^3x_2&0&h_{410}&h_{320}&h_{510}&h_{420}&h_{330}&h_{610}&h_{520}&h_{430}&h_{340}&h_{710}&h_{620}&h_{530}&h_{440}&h_{350}\\
          x_1^2x_2^2&\frac{49}{6}&h_{320}&h_{230}&h_{420}&h_{330}&h_{240}&h_{520}&h_{430}&h_{340}&h_{250}&h_{620}&h_{530}&h_{440}&h_{350}&h_{260}\\
          x_1x_2^3&0&h_{230}&h_{140}&h_{330}&h_{240}&h_{150}&h_{430}&h_{340}&h_{250}&h_{160}&h_{530}&h_{440}&h_{350}&h_{260}&h_{170}\\
          x_2^4&0&h_{140}&h_{050}&h_{240}&h_{150}&h_{060}&h_{340}&h_{250}&h_{160}&h_{070}&h_{440}&h_{350}&h_{260}&h_{170}&h_{080}
        \end {array} \right] 
    \end{displaymath}
  }
  
\item Extract a principal minor of full rank.
  \EXAMPLE{
    In our example the $6 \times 6$ principal minor is of full rank.
    The matrix $\Delta_0$ is 
    \begin{displaymath}
      \Delta_0 = 
      \left[ 
        \begin {array}{cccccc} 
          0&{\frac{57}{4}}&0&0&0&{\frac{28}{3}}\\
          {\frac{57}{4}}&0&0&{\frac{79}{4}}&0&\frac{1}{3} \\
          0&0&{\frac{28}{3}}&0&\frac{1}{3}&0\\
          0&{\frac{79}{4}}&0&0&0&{\frac{49}{6}}\\
          0&0&\frac{1}{3}&0&{\frac{49}{6}}&0\\
          {\frac{28}{3}}&\frac{1}{3}&0&{\frac{49}{6}}&0&0
        \end {array} 
      \right] 
    \end{displaymath}
    The columns (and the rows) of the matrix correspond to the monomials
    $\{ 1, x_1, x_2, x_1^2, x_1x_2, x_2^2\}$.
  }

\item We compute the ``shifted'' matrix $\Delta_1 = x_1 \Delta_0$.
  \EXAMPLE{
    The shifted matrix $\Delta_1$ is 
    \begin{displaymath}      
      \Delta_1 = 
      \left[ 
        \begin {array}{cccccc} 
          {\frac{57}{4}}&0&0&{\frac{79}{4}}&0&\frac{1}{3}\\
          0&{\frac{79}{4}}&0&0&0&{\frac{49}{6}}\\
          0&0&\frac{1}{3}&0&{\frac{49}{6}}&0\\
          {\frac{79}{4}}&0&0&h_{{500}}&h_{{410}}&h_{{320}}\\
          0&0&{\frac{49}{6}}&h_{{410}}&h_{{320}}&h_{{230}}\\
          \frac{1}{3}&{\frac{49}{6}}&0&h_{{320}}&h_{{230}}&h_{{140}}
        \end {array} 
      \right] 
    \end{displaymath}
    The columns of the matrix correspond to the monomials
    $\{ x_1, x_1^2, x_1x_2, x_1^3, x_1^2x_2, x_1x_2^2\}$
    which are the monomials that correspond to the columns of $\Delta_0$,
    i.e. $\{ 1, x_1, x_2, x_1^2, x_1x_2, x_2^2\}$, multiplied by $x_1$.

    Since not all the entries of $\Delta_1$ are known, we need to compute them in order to proceed. 
  } 
  
  Consider the following method to extend the matrix of a quotient algebra.
  In the quotient algebra it holds that 
  $\mat{M}_{x_i} \mat{M}_{x_j} - \mat{M}_{x_j} \mat{M}_{x_i} = \mat{O}$ \cite{mourrain-aaecc-1999}
  for any $i, j \in \{1, \dots, n\}$,
  i.e. the matrices of multiplications commute (cf. Section \ref{sec:trunc:hankel}).

  From Proposition~\ref{prop:Hl-mx} we know that 
  $\Delta_i = \mat{M}_{x_i}^{t} \Delta_0$,
  and hence $\mat{M}_{x_i}^{t} = \Delta_i \Delta_0^{-1}$,
  for $1 \leq i \leq n$.
  
  We form all the possible matrix equations,
  $\mat{M}_{x_i} \mat{M}_{x_j} - \mat{M}_{x_j} \mat{M}_{x_i} = \mat{O}$,
  there are ${n \choose 2}$, and we equate their elements to zero.
  Since the dimension of the matrices is $r \times r$,
  this leads to at most ${n \choose 2} r^2$, 
  or $\OO(n^2 r^2)$ equations.
  Note that the equations are, at most of total degree 2.

  \EXAMPLE{
    In our example the matrix $\Delta_2$ is
    \begin{displaymath}      
      \Delta_2 = 
      \left[ 
        \begin {array}{cccccc} 
          0&0&{\frac{28}{3}}&0&\frac{1}{3}&0\\
          0&0&\frac{1}{3}&0&{\frac{49}{6}}&0\\
          {\frac{28}{3}}&\frac{1}{3}&0&{\frac{49}{6}}&0&0\\
          0&0&{\frac{49}{6}}&h_{{410}}&h_{{320}}&h_{{230}}\\
          \frac{1}{3}&{\frac{49}{6}}&0&h_{{320}}&h_{{230}}&h_{{140}}\\
          0&0&0&h_{{230}}&h_{{140}}&h_{{050}}
        \end {array} 
      \right] 
    \end{displaymath}

    Since we have only two variables, there is only one matrix equation,
    \begin{displaymath}
      \mat{M}_{x_i} \mat{M}_{x_j} - \mat{M}_{x_j} \mat{M}_{x_i} = 
      \Delta_1 \Delta_0^{-1} \Delta_2 \Delta_0^{-1}
      - \Delta_2 \Delta_0^{-1} \Delta_1 \Delta_0^{-1}
      = \mat{O}.      
    \end{displaymath}
    
    Many of the resuling equations are trivial. 
    After disgarding them, we have 6 unknonws
    $\set{ h_{500},h_{410},h_{320},h_{230},h_{140},h_{050} }$
    and 15 equations.

    A solution of the system is the following
    $$\set{ h_{500}=1, h_{410}=2, h_{320}=3,
      h_{230}= 1.5060, h_{140}= 4.960,h_{050}= 0.056}.$$
    
    We subsitute these values to $\Delta_1$ and we continue the algorithm as in the
    previous example.
  }

  Other algorithms to extend a moment matrix,
  e.g. \cite{laurent-ima-2008,laurent-ams-2005,cf-hjm-1991},
  the so called {\em flat estensions}, are applicable when the $\Delta_0$ is
  positive definite.

\item We solve the equation $(\Delta_1 - \lambda \Delta_0) X = 0$.
  \EXAMPLE{
    The normalized eigenvectors of the generalized eigenvalue problem are
    \begin{displaymath}
      \left[ \begin {array}{c}  
          1\\
          -0.830+1.593\,i\\
          -0.326-0.0501\,i\\ 
          -1.849-2.645\,i\\
          0.350-0.478\,i\\
          0.103+0.0327\,i
        \end {array} \right] , 
      \left[ \begin {array}{c}  
          1 \\
          -0.830 -1.593\,i\\ 
          -0.326 +0.050\,i\\
          -1.849 +2.645\,i\\
          0.350 +0.478\,i\\
          0.103 -0.032\,i\\
        \end {array} \right] , 
      \left[ \begin {array}{c}  
          1.0 \\
          1.142\\
          0.836\\ 
          1.305\\ 
          0.955\\ 
          0.699
        \end {array} \right],
    \end{displaymath}
    \begin{displaymath}
      \left[ \begin {array}{c} 
          1
          0.956\\-
          0.713\\ 
          0.914\\ 
          -0.682\\
          0.509
        \end {array} \right] ,
      \left[ \begin {array}{c}  
          1 \\
          -0.838 + 0.130\,i\\ 
          0.060+ 0.736\,i\\
          0.686-0.219\,i\\
          -0.147-0.610\,i\\
          -0.539+ 0.089\,i
        \end {array} \right] , 
      \left[ \begin {array}{c} 
          1 \\        
          - 0.838- 0.130\,i\\
          0.060- 0.736\,i \\
          0.686+ 0.219\,i\\
          - 0.147+ 0.610\,i\\
          - 0.539- 0.089\,i
        \end {array} \right] . 
    \end{displaymath}

    The coordinates of the eigenvectors correspond to the elements
    $\{ 1, x_1, x_2, x_1^2, x_1x_2, x_2^2\}$
    and we can recover the coefficients of $x_1$ and $x_2$ 
    in the decomposition.

    After, solving the over-constrained linear system for the coefficients of the linear
    forms we deduce the decomposition
    \begin{eqnarray*}
      (0.517+0.044\,i)  \left(x_{0}-(0.830-1.593\,i) x_{1} -(0.326+0.050\,i)
        x_{2}\right)^{4} \\
      +(0.517-0.044\,i) \left(x_{0}-(0.830+1.593\,i) x_{1} -(0.326-0.050\,i) x_{2}\right)^{4} \\
      +2.958 \left(x_{0}+(1.142)x_{1}+0.836x_{2}\right)^{4}\\
      +4.583 \left(x_{0}+(0.956)x_{1}-0.713x_{2}\right)^{4} \\
      -(4.288+1.119\,i)\left(x_{0}-(0.838-0.130\,i)x_{1}+(0.060+0.736\,i)x_{2}\right)^{4}\\
      -(4.288-1.119\,i)\left(x_{0}-(0.838+0.130\,i)x_{1}+(0.060-0.736\,i)x_{2}\right)^{4}
    \end{eqnarray*}
  }
\end{enumerate}

\section{Conclusions and future work}

We propose an algorithm that computes symmetric tensor decompositions, extending Sylvester's algorithm.
The main ingredients are 
i) reformulate the problem in a dual space,
ii) exploit the properties of multivariate Hankel operators and Gorenstein algebra,
iii) devise an effective method to solve, when necessary, the truncated Hankel problem,
iv) deduce the decomposition by solving a generalized eigenvalue problem.

There are several open questions that we are currently working on.
What is the (arithmetic and Boolean) complexity of the algorithm?
If we do not know all the the elements of the tensor, can we still compute a decomposition?

\subsection*{Acknowledgments.} 
This work is partially supported by contract ANR-06-BLAN-0074 "Decotes".


\appendix

\newpage

\section{Ternary cubics}

As an application, we present the decomposition of all the types of ternary cubics.
The decomposition allows us to classify, up to projective transformations of the
variables, homogeneous polynomials of  degree three in three variables, for instance with the help of the algorithm described in  \cite{cm-sp-1996}. 
For another algorithm for decomposing ternary cubics, based on the method of moving
frames and on triangular decompositions of algebraic varieties, we refer the
reader to \cite{KogMaz-issac-2002}. Two polynomial are equivalent in this
classicifation if there exists a varaibles invertible trnasfmation which maps
one polynomial to the other.

The classification algorithm goes as follows. Given a ternary cubic, we
compute its decomposition as a sum of powers of linear forms.
We have the following cases:
\begin{itemize}
\item  If the rank is one then the polynomial is a $3^{\mathrm{rd}}$
  power of a linear form,
  that is, it is equivalent to $x_0^3$.
\item If the rank is 2, then the polynomial is equivalent to $x_{0}^{3}+
  x_{1}^{3}$ and is in the orbit 
  of $x_0 x_1 (x_0 + x_1)$. In fact, the decomposition of the latter polynomial is 
  \begin{eqnarray*}
    \left( -260712-{\frac {1628000}{9}}\,i\sqrt {3} \right)  
    \left(  \left( -{\frac {125}{16492}}-{\frac {17}{16492}}\,i\sqrt {3} \right) 
      x_{{0}}+ \left( -{\frac {22}{4123}}+{\frac {27}{8246}}\,i\sqrt {3} \right) x_{{1}} \right)^{3} \\
    + \left( -260712+{\frac {1628000}{9}}\,i\sqrt {3} \right)  
    \left(  \left( -{\frac {125}{16492}}+{\frac {17}{16492}}\,i\sqrt {3} \right) x_{{0}}
      + \left( -{\frac {22}{4123}}-{\frac {27}{8246}}\,i\sqrt {3} \right) x_{{1}} \right) ^{3}.
  \end{eqnarray*}
\item  If the rank is 3, then the polynomial is either in the orbit of
  $x_0^2 x_1$ or in the orbit of $x_0^3  + x_1^3  + x_2^3$.
  To identify the orbit, it suffice to check if the polynomial is square-free or
  not (that is, check whether the gcd between the polynomial and one of its derivatives is 1). If it is not square-free then it is in the orbit of $x_0^2 x_1$.
  Otherwise it is in the orbit of  $x_0^3  + x_1^3  + x_2^3$.
  
  The decomposition of $x_0^2 x_1$ is 
  \begin{displaymath}
    \begin{array}{r}
      -0.15563-0.54280\, i ( (-0.95437-0.08483\, i) x_0 + (-0.00212+0.28631\, i) x_1)^3 \\ 
      -0.15723+0.54841\, i ( (-0.95111+0.09194\, i) x_0 + (-0.00222-0.29484\, i) x_1)^3 \\ 
      -0.04320 ( (-0.49451) x_0 + (-0.86917) x_1)^3 .
    \end{array}
  \end{displaymath}
  %
  %
\item If the rank is 4, then our polynomial is generic.
  As an example, consider the  polynomial
  $150\,{x_{0}}^{2}x_{2}+{x_{1}}^{2}x_{2}+{x_{2}}^{3}-12\,{x_{0}}^{3}$;
  a decomposition of which is
  \begin{displaymath}
    \begin{array}{r}
      0.53629 ( +0.34496 x_0  +0.71403 x_1  +0.60923 x_2 )^3 \\ 
      -195.64389 ( -0.99227 x_0  +0.00286 x_1  -0.12403 x_2 )^3 \\ 
      +211.45588 ( -0.99282 x_0  +0.00311 x_1  +0.11962 x_2 )^3 \\ 
      +0.52875 ( -0.34600 x_0  -0.71671 x_1  +0.60549 x_2 )^3.
    \end{array}
  \end{displaymath}

\item If the rank is 5, then the polynomial is of maximal rank and it is in the
  orbit of $x_0^2 x_1 + x_0 x_2^2$, a decomposition of which, is 
  \begin{eqnarray*}
    +0.28100 ( +0.06322 x_0  -0.99748 x_1  +0.03224 x_2 )^3 \\ 
    +0.97839 ( +0.14391 x_0  +0.50613 x_1  +0.85036 x_2 )^3 \\ 
    +0.44877 ( +0.73493 x_0  +0.56369 x_1  -0.37702 x_2 )^3 \\ 
    +(-0.97396-0.94535\, i) ( 0.45304 x_0 + (-0.60752+0.14316\, i) x_1 + (-0.52915+0.35382\, i) x_2 )^3 \\ 
    +(-0.97396+0.94535\, i) ( 0.45304 x_0 + (-0.60752+0.14316\, i) x_1 + (-0.52915-0.35382\, i) x_2 )^3 .
  \end{eqnarray*}

  %
\end{itemize}

\section{An example of extreme rank}
\label{sec:ex-max-rank}

In this section we present in detail the decomposition of a ternary cubic of
maximal rank, that is 5. Consider the polynomial 
\begin{displaymath}
  x_0^2 x_1 + x_0 x_2^2.
\end{displaymath}

The matrix of the quotient algebra is
\begin{displaymath}
  \left[ \begin {array}{cccccccccc} 
      0&\frac{1}{3}&0&0&0&\frac{1}{3}&0&0&0&0\\
      \frac{1}{3}&0&0&0&0&0&h_{4,0,0}&h_{3,1,0}&h_{2,2,0}&h_{1,3,0}\\
      0&0&\frac{1}{3}&0&0&0&h_{3,1,0}&h_{2,2,0}&h_{1,3,0}&h_{0,4,0}\\
      0&0&0&h_{4,0,0}&h_{3,1,0}&h_{2,2,0}&h_{5,0,0}&h_{4,1,0}&h_{3,2,0}&h_{2,3,0}\\
      0&0&0&h_{3,1,0}&h_{2,2,0}&h_{1,3,0}&h_{4,1,0}&h_{3,2,0}&h_{2,3,0}&h_{1,4,0}\\
      \frac{1}{3}&0&0&h_{2,2,0}&h_{1,3,0}&h_{0,4,0}&h_{3,2,0}&h_{2,3,0}&h_{1,4,0}&h_{0,5,0}\\
      0&h_{4,0,0}&h_{3,1,0}&h_{5,0,0}&h_{4,1,0}&h_{3,2,0}&h_{6,0,0}&h_{5,1,0}&h_{4,2,0}&h_{3,3,0}\\
      0&h_{3,1,0}&h_{2,2,0}&h_{4,1,0}&h_{3,2,0}&h_{2,3,0}&h_{5,1,0}&h_{4,2,0}&h_{3,3,0}&h_{2,4,0}\\
      0&h_{2,2,0}&h_{1,3,0}&h_{3,2,0}&h_{2,3,0}&h_{1,4,0}&h_{4,2,0}&h_{3,3,0}&h_{2,4,0}&h_{1,5,0}\\
      0&h_{1,3,0}&h_{0,4,0}&h_{2,3,0}&h_{1,4,0}&h_{0,5,0}&h_{3,3,0}&h_{2,4,0}&h_{1,5,0}&h_{0,6,0}
    \end {array} 
  \right],
\end{displaymath}
and the matrices $\Delta_0$, $\Delta_1$ and $\Delta_2$ are
\begin{displaymath}
  \left[ 
    \begin {array}{ccccc} 
      0&1/3&0&0&0\\
      1/3&0&0&0&0\\
      0&0&1/3&0&0\\
      0&0&0&h_{4,0,0}&h_{3,1,0}\\
      0&0&0&h_{3,1,0}&h_{2,2,0}
    \end {array} 
  \right],
  \left[ 
    \begin {array}{ccccc} 
      1/3&0&0&0&0\\
      0&0&0&h_{4,0,0}&h_{3,1,0}\\
      0&0&0&h_{3,1,0}&h_{2,2,0}\\
      0&h_{4,0,0}&h_{3,1,0}&h_{5,0,0}&h_{4,1,0}\\
      0&h_{3,1,0}&h_{2,2,0}&h_{4,1,0}&h_{3,2,0}
    \end {array} 
  \right] ,
  \left[ 
    \begin {array}{ccccc} 
      0&0&1/3&0&0\\
      0&0&0&h_{3,1,0}&h_{2,2,0}\\
      1/3&0&0&h_{2,2,0}&h_{1,3,0}\\
      0&h_{3,1,0}&h_{2,2,0}&h_{4,1,0}&h_{3,2,0}\\
      0&h_{2,2,0}&h_{1,3,0}&h_{3,2,0}&h_{2,3,0}
    \end {array} 
  \right] .
\end{displaymath}

If we form the matrix equation 
\begin{displaymath}
  \mat{M}_{x_i} \mat{M}_{x_j} - \mat{M}_{x_j} \mat{M}_{x_i} = 
  \Delta_1 \Delta_0^{-1} \Delta_2 \Delta_0^{-1}
  - \Delta_2 \Delta_0^{-1} \Delta_1 \Delta_0^{-1}
  = \mat{O}.      
\end{displaymath}
then we have a system of 8 equations in 8 unknowns.
The unknowns are  
$$\left\{
  h_{{5,0,0}},h_{{4,1,0}},h_{{4,0,0}},h_{{3,1,0}},h_{{2,2,0}},h_{{1,3,0}},h_{{3,2,0}},h_{{2,3,0}}
\right\}.$$
It turns out that the system is not zero dimensional, and that we can choose (randomly)
the values of five of them, i.e.
$\set{  h_{{1,3,0}}=3,h_{{3,1,0}}=1,h_{{2,2,0}}=2,h_{{4,1,0}}=4,h_{{4,0,0}}=5 }$.
Working as in the other examples we end up with the decomposition

\begin{eqnarray*}
  +  0.000071( x_{0} -15.778 x_{1} + 0.510 x_{2} )^3\\
  + 0.002916( x_{0} + 3.517 x_{1} + 5.909 x_{2} )^3\\
  + 0.178137( x_{0} + 0.767 x_{1} - 0.513 x_{2} )^3 \\
  (-0.09056 -0.0879\,i)( x_{0} + (-1.341+0.316\,i )x_{1} +(-1.168+0.781\,i)x_{2})^3 \\
  (-0.09056 +0.0879\,i)( x_{0} + (-1.341+0.316\,i )x_{1} +(-1.168-0.781\,i)x_{2})^3.
\end{eqnarray*}


\end{document}